%% file: arXiv_v1.tex
\documentclass[a4paper,allowtoday,twocolumn,unpublished]{quantumarticle} 
\pdfoutput=1 

\usepackage{dsfont} 	
\usepackage{microtype} 	
\usepackage[numbers,sort&compress]{natbib} 	
\usepackage[utf8]{inputenc} 	
\usepackage[T1]{fontenc} 	
\usepackage[british]{babel} 	
\usepackage[dvipsnames,x11names]{xcolor} 	
\usepackage{amsmath,amssymb,amsthm,bm,amsfonts,bbm} 	
\usepackage{mathtools} 	
\usepackage{physics} 	
\usepackage[unicode,colorlinks=true,citecolor=OliveGreen,linkcolor=Purple,urlcolor=Magenta]{hyperref} 	
\usepackage{thmtools}
\usepackage{thm-restate}
\usepackage{caption} 
\usepackage{subcaption}
\usepackage{float}

\input{0_MTQ_macros}  

\newtheorem{corollary}{Corollary}

\newcommand{\hilbert}[1]{\mathcal{H}_{#1}}

\begin{document}
\title{Entanglement in quantum channel discrimination: sometimes less is more}

\author{Kristin Sundal Lien}
\affiliation{Sorbonne Universit\'{e}, CNRS, LIP6, F-75005 Paris, France}
\affiliation{Department of Physics, NTNU - Norwegian University of Science and Technology, Trondheim, Norway}

\author{Marco Túlio Quintino}
\orcid{0000-0003-1332-3477}
\affiliation{Sorbonne Universit\'{e}, CNRS, LIP6, F-75005 Paris, France}
\date{\today}
\begin{abstract}
Entanglement is known to be a powerful resource that improves performance in various quantum information and computational tasks. A standard example of such a phenomenon is the possibility of perfectly discriminating all four Pauli operations in a single shot via the superdense coding protocol. While entanglement is often a powerful resource for quantum channel discrimination, this is not necessarily the case. In this work, we identify scenarios in which the maximally entangled state is a bad choice of input state and, more generally, show that excessive entanglement can reduce channel discriminability dramatically. To do so, we present an explicit pair of unitary channels which are perfectly discriminable without entanglement, but for which any strategy with maximally entangled input states is $\epsilon$-close to a blind uniform guessing strategy. To develop a systematic approach, we introduce the concepts of Maximal Entanglement Worst Case (MEWC) and Maximal Entanglement Best Case (MEBC) pairs of channels, and present conditions for a pair of channels to be MEWC or MEBC. With these conditions, we show that the optimal input states for discriminating MEWC pairs of channels are necessarily separable, and provide non-trivial examples of measurement channels for which entanglement necessarily reduces the maximum probability of discrimination.
\end{abstract}

\maketitle


\tableofcontents

\section{Introduction}
A quantum channel is a general description deterministic admissible transformations of one quantum state into another, and it thus captures the effect of a physical process on a quantum system~\cite{watrous_theory_2018}. The ability to, as accurately as possible, describe and distinguish between different physical processes, is a fundamental part of physics. In a broad sense, physics is based on performing measurements on a system with the goal of understanding or confirming what theory correctly predicts the processes it has been subject to.

Quantum channel discrimination is a fundamental task that formalises the idea of distinguishing between possible physical processes. Being closely related to (and developed through the study of) parameter estimation and hypothesis testing~\cite{helstrom_minimum_1967, helstrom_quantum_1969}, quantum channel discrimination falls within the field of quantum metrology~\cite{giovannetti_quantum-enhanced_2004, toth_quantum_2014}. Some questions regarding quantum channel discrimination that have been studied include deriving fundamental bounds for how well channels can be discriminated~\cite{pirandola_fundamental_2019}, in what cases perfect discrimination is possible~\cite{datta_perfect_2021, acin_statistical_2001}, optimal input states~\cite{jencova_conditions_2016, caiaffa_channel_2018}, memory effects in channel discrimination~\cite{chiribella_memory_2008, ohst_characterising_2024}, and the use of parallel, sequential and indefinite-causal-order strategies for channel discrimination when multiple uses of the channel are available~\cite{bavaresco_strict_2021, bavaresco_unitary_2022}. 

Here, we consider the task of single-shot minimum error quantum channel discrimination, which consists in finding a strategy that maximises the success probability of determining  which one, among a known set of quantum channels, a device produces, while only having access to one use of the device~\cite{puzzuoli_entanglement_2018}. In particular, we will study the role of entanglement in this task. Many works highlight the advantages that can be achieved by using an entangled input state for channel discrimination~\cite{jencova_conditions_2016,puzzuoli_entanglement_2018, dariano_using_2001, sacchi_optimal_2005, sacchi_entanglement_2005, sacchi_minimum_2005, dariano_improved_2002}.

Furthermore, it has been shown that every entangled input state provides an advantage in some channel discrimination task~\cite{piani_all_2009}, and that more entanglement provides a greater advantage in the sense that there exists a larger set of channel pairs for which the state outperforms any less entangled state~\cite{bae_more_2019}. Building on this notion of usefulness of an input state,~\cite{caiaffa_channel_2018} defines the term channel discrimination power of a bipartite state, and show that the maximally entangled input state has the highest discriminative power in this sense. Additionally, there are several classes of channels for which it has been shown that the maximally entangled input state is always optimal~\cite{jencova_conditions_2016}, including all pairs of unital qubit channels, amongst which we find many important channel classes, like qubit unitaries, qubit projective measurements and the Pauli channels. 

However, an entangled input state does not always improve distinguishability between quantum channels. In~\cite{dariano_improved_2002}, it is shown that entanglement cannot improve distinguishability between two arbitrary unitary quantum channels. In particular, there are several known examples of channel discrimination tasks where the maximally entangled input state performs strictly worse than a partially entangled input state, as demonstrated in~\cite{jencova_conditions_2016, manna_maximally_2025}. These are the situations that we are interested in exploring and characterising through this work. Whereas the maximally entangled input state performs well in an arbitrary discrimination task, seeing as it is the ``least bad'' input in a worst case scenario, it is not necessarily a good input state for a specific, known channel discrimination task. 

In this work, we address the question of when having maximally entangled or highly entangled input states is a bad choice for channel discrimination. For this goal, we introduce the concepts of Maximal Entanglement Worst Case (MEWC) and Maximal Entanglement Best Case (MEBC) pairs of channels, and prove conditions for a pair of channels to be MEWC or MEBC. We then make use of our theorems to construct explicit examples where maximally entangled states or highly entangled states attain a poor performance when compared to separable or less entangled states.

This paper is structured as follows. We begin by presenting some motivating examples that demonstrate how the entanglement of the input state affects discriminability. In Section~\ref{ch:prelim}, we introduce our notation and some mathematical preliminaries. Thereafter, in Section~\ref{ch:results}, we present our main results regarding the role of entanglement in quantum channel discrimination tasks. We finish by demonstrating some examples of discrimination tasks where entanglement reduces discriminability in Sections~\ref{ch:measurement_ch} and~\ref{ch:unitary_ch}. 

\section{Examples: The maximal success probability as a function of the entanglement}
\label{ch:motivation_ex}
In this section, we present some discrimination tasks that demonstrate the effect that the amount of entanglement of the input state has. We will begin with two examples where entanglement increases the distinguishability between the channels, exemplifying why entanglement is often considered a resource for quantum channel discrimination. Thereafter, we show two examples of discrimination task where the optimal strategy uses no entanglement, and where, in fact, any entanglement strictly reduces discriminability. These examples demonstrate the main message of our work: that entanglement may act as a disturbance rather than as a resource in some quantum channel discrimination tasks. 

\subsection{Entanglement may be good for channel discrimination}
\label{ch:motivation_ex_ME_bad}
\subsubsection{Discriminating four Pauli channels}
An example that illustrates the power of entanglement for channel discrimination, is the task of discriminating the four Pauli channels, described by the Pauli matrices $\{ X, \, Y,\, Z,\, I \}$, with uniform prior probabilities $p = \frac{1}{4}$ for each of the four channels. If we are restricted to performing the Pauli operations on a single qubit $\ket{\psi}\in\mathbb{C}^2$, the task of discriminating the four Pauli channels is equivalent to discriminating among the four qubit states  $\{ X\ket{\psi},\, Y \ket{\psi}, \, Z\ket{\psi}, \, I\ket{\psi}\}$. The maximal success probability is then limited to $\frac{2}{4}=\frac{1}{2}$. If we, on the other hand, use a maximally entangled input state $\ket{\phi^+}:=\frac{\ket{00}+\ket{11}}{\sqrt{2}}$ and apply the channel to half of it, we want to discriminate the states
\begin{align}
    (I\otimes X)\ket{\phi^+} = \ket{\psi^+}& \\
    (I\otimes Y)\ket{\phi^+} = \ket{\psi^-}& \\
    (I\otimes Z)\ket{\phi^+} = \ket{\phi^-}& \\
    (I\otimes I)\ket{\phi^+} = \ket{\phi^+}& ,
\end{align}
which are mutually orthogonal and therefore perfectly distinguishable. This is exactly the property exploited in the superdense coding protocol~\cite{nielsen_quantum_2010, bennett_communication_1992}, which is often used as an example of how entanglement can be used as a resource. 

In order to have a quantitative analysis of the usefulness of entanglement, here we present the maximal success probability of discriminating the four Pauli channels as a function of the entanglement of the two-qubit input state $\ket{\psi}\in\mathbb{C}^2\otimes\mathbb{C}^2$ quantified by the entanglement entropy\footnote{The entanglement entropy of a bipartite pure state $\ket{\psi_{AA'}}\in \mathcal{H}_A\otimes\mathcal{H}_{A'}$  is defined as the entropy of $\tr_{A'}(\ketbra{\psi_{AA'}})$. }. A plot\footnote{Details on how these plots are produced are presented in the Appendix~\ref{ch:appendix}.} of the maximal success probability for discriminating the four Pauli channels as a function of the entanglement entropy is presented in Figure~\ref{fig:pauli_plot}. We can see that the success probability ranges from $p^*=0.5$, without entanglement, and increases until perfect discrimination $p^*=1$ with maximal entanglement.  

\begin{figure}[H]
    \centering
        \includegraphics[width=\columnwidth]{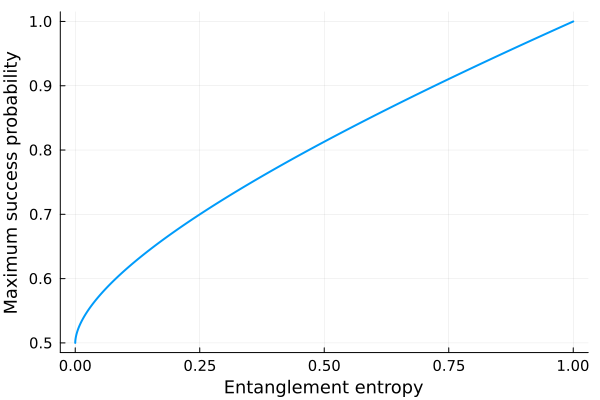}
        \caption{Maximal success probability of discrimination between the four Pauli channels. Success probability of discrimination increases with increasing entanglement.}
        \label{fig:pauli_plot}
\end{figure}

For dimension $d>2$, the Pauli unitaries can be generalised as the Clock-Shift operators%
\footnote{The clock-shift operators are given by $X_d^i Z_d^j, \quad \text{for } i,j \in \{0,\dots,d-1\}$, where $X_d := \sum_{l=0}^{d-1} \ketbra{ l \oplus 1}{l} , \qquad Z_d := \sum_{l=0}^{d-1} \omega^l \ketbra{l}{l}$, and $ \omega := e^{2\pi i / d}$.}
, for which it is known that the maximal probability of successful discrimination depends on the dimension of the auxiliary space, $d_{A'}$ via $p^*  = \min \{1,\frac{d_{A'}}{d}\}$ (see e.g.~\cite{ohst_characterising_2024} for a proof). We notice that for large dimensions $d$ and with separable input states $d_{A'} = 1$, the success probability tends to zero, whereas for $d_{A'} = d$, the maximal success probability is one, supporting the fact that entanglement is useful in this discrimination task. 

\subsubsection{Discriminating two Werner-Holevo channels}
We now consider discrimination between the symmetric and the antisymmetric Werner-Holevo channel, defined as follows~\cite{werner_counterexample_2002}:
\begin{align}
    \Phi_{\pm}(\rho) = \frac{1}{d\pm1} \Tr(\rho)I \pm \rho^{T}. 
\end{align}
When applied to one part of a maximally entangled state, these channels give
\begin{align}
    I \otimes \Phi_{\pm} (\ketbra{\phi^+}) & = \frac{1}{d} \sum_{i,j = 0}^{d-1} \ketbra{i}{j} \otimes \Phi_{\pm} (\ketbra{i}{j})\\
    & = \frac{1}{d(d \pm 1)} \left[ I  \pm F \right], 
\end{align}
where $F := \sum_{i,j =0}^{d-1} \ketbra{i}{j} \otimes \ketbra{j}{i}$ is the SWAP operator. 

We can check that the states $ I \otimes \Phi_{\pm} (\ketbra{\phi^+})$ have orthogonal support\footnote{We may also recognise $\Pi_\text{sym}:= \frac{I + F}{(d + 1)} $ as the projector onto the symmetric subspace and  $\Pi_\text{asym}:= \frac{I  - F}{(d - 1)}$ as the projector onto the antisymmetric subspace.}, hence they are perfectly discriminable. Therefore, a maximally entangled input state is optimal as pointed out in~\cite{jencova_conditions_2016, puzzuoli_ancilla_2017}. Furthermore, Figure~\ref{fig:wh_plot} shows that decreasing the amount of entanglement strictly decreases the success probability of discrimination. 

\begin{figure}[H]
    \centering
        \includegraphics[width=\columnwidth]{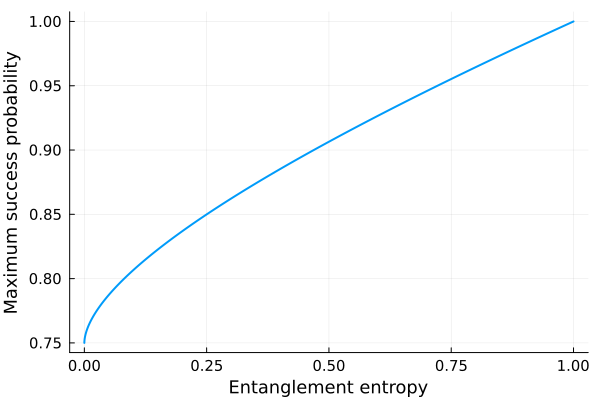}
        \caption{Maximal success probability of discrimination between Werner-Holevo channels. Success probability of discrimination increases with increasing entanglement.}
        \label{fig:wh_plot}
\end{figure}

\subsection{Entanglement may be bad for channel discrimination}
 The examples presented so far support the fact that entanglement is a powerful resource in discrimination tasks. However, this is not always the case. We now present two channel discrimination tasks where an entangled input state is not optimal for discrimination.
 
 \subsubsection{Discriminating two measurement channels}
 Figure~\ref{fig:m_plot} shows a discrimination task where increasing entanglement strictly decreases discrimination. 
\begin{figure}[H]
    \centering
        \includegraphics[width=\columnwidth]{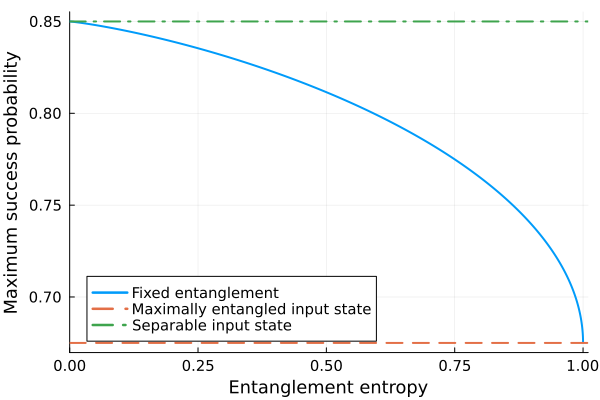}
        \caption{Maximal success probability of discrimination between two measurement channels. Success probability of discrimination decreases with increasing entanglement. The top horizontal (green) dot-dashed line is the maximal success probability with separable states, the bottom horizontal (orange) dashed line is the maximal success probability with separable states, and the (blue) filled curve is the optimal success probability with a fixed value of entanglement entropy.}
        \label{fig:m_plot}
\end{figure}
The channels discriminated in Figure~\ref{fig:m_plot} are two measurement channels corresponding to dichotomic measurements given as follows: 
\begin{align}
        \label{eq:measurement_channel}
        \Phi_{M^{X}}(\rho)  = \sum_{i=0}^{1} \Tr(\rho M^{X}_i)\, \ketbra{i}, \, X \in (A,B) , 
\end{align}
where 
\begin{align}
    M^{A}_0 = 
    \begin{pmatrix}
            \frac{1}{10} &0\\
            0& \frac{1}{5}
        \end{pmatrix}, \, M^{A}_1 = I-M^{A}_0
\end{align}
and 
\begin{align}
    M^{B}_0 = 
    \begin{pmatrix}
            \frac{3}{10} &\sqrt{\frac{1}{10}}\\
            \sqrt{\frac{1}{10}} & \frac{7}{10}
        \end{pmatrix}, \, M^{B}_1 = I-M^{B}_0. 
\end{align}
We see that the success probability of discrimination is optimal for a separable input state and drops towards its minimum for a maximally entangled input state. This shows that, for some channels, entangled input states have a strictly worse performance then separable ones, and that the use of entanglement may reduce the performance in some channel discrimination tasks. The details of this example and this class of channels will be discussed in section~\ref{ch:measurement_ch}. 

\subsubsection{Discriminating two unitary channels}
In this final example, we present a pair of unitary channels where a maximally entangled state is a bad choice of input state. Consider a pair of qudit unitaries given by identity $I$ and $V$ as follows:  
\begin{align}
    I:=& \sum_{i=0}^{d-1} \ketbra{i}, \\
    V:=& \ketbra{0} - \sum_{i=1}^{d-1}\ketbra{i}.
\end{align}
$I$ and $V$ are perfectly distinguishable with an unentangled input state. However, a maximally entangled input state gives a success probability equivalent to random guessing ($p^*=0.5$) in the limit $d\to \infty$ and is therefore asymptotically useless. 

The strategies in the two cases can be understood as follows. If we input the maximally entangled state $ \ket{\phi^+} := \frac{1}{\sqrt{d}}\sum_{i=0}^{d-1} \ket{i}\ket{i}$, we want to distinguish the output states $ I \otimes I  \ket{\phi^+} =  \ket{\phi^+}$ and $ W \otimes I  \ket{\phi^+} = \frac{1}{\sqrt{d}} \left(\ket{0}\ket{0} - \sum_{i=1}^{d-1} \ket{i}\ket{i} \right)$. As $d$ increases, the two states approach one another (when disregarding the global phase of $-1$) and become increasingly indistinguishable. 

If we, on the other hand, choose the unentangled input state $\ket{+}  := \frac{\ket{0}+ \ket{1}}{\sqrt{2}}$, we want to distinguish the states $I \ket{+}  = \ket{+}$ and $ W  \ket{+}  =   \frac{\ket{0} - \ket{1}}{\sqrt{2}} := \ket{-}$, which are orthogonal and thus perfectly distinguishable in any number of dimensions $d$.

Albeit simple, this example demonstrates a unitary discrimination task where the maximally entangled state is not a good choice of input state. In the limit of infinite dimension, it is even completely useless. 

We have now seen some examples showing that while quantum entanglement is often viewed as a resource for quantum channel discrimination, in particular tasks, the use of entanglement may instead reduce the dicriminability. In the following, we will analyse this problem from an analytical perspective. 
\section{Notation and preliminaries}
\label{ch:prelim}
Quantum systems are denoted by capital letters $A,B,...$, and they are associated to Hilbert spaces $\mathcal{H}_A \cong \mathbb{C}^{d_A}$, where $d_A$ is the dimension of the Hilbert space $\mathcal{H}_A$. The space of linear maps between $\mathcal{H}_A$ and $\mathcal{H}_B$ is denoted $\mathcal{L}(\mathcal{H}_A, \mathcal{H}_B)$, and the linear operators on $\mathcal{H}_A$ given by $ \mathcal{L}(\mathcal{H}_A, \mathcal{H}_A)$ is denoted $  \mathcal{L}(\mathcal{H}_A)$ for brevity. The notation $\mathcal{D(\hilbert{A})}$ is used for density operators on $\hilbert{A}$.  

Quantum channels are completely positive and trace preserving (CPTP) mappings of the form 
\begin{align}
    \Phi : \mathcal{L}(\hilbert{A}) \to \mathcal{L}(\hilbert{B}), 
\end{align}
transforming valid density operators to valid density operators. They are also called superoperators, and can be represented as operators through the Choi-Jamiołkowski isomorphism: 
\begin{definition}[The Choi-Jamiołkowski isomorphism~\cite{taranto_higher-order_2025, choi_completely_1975}]
Let $\mathcal{S}: \mathcal{L}(\hilbert{A}) \to \mathcal{L}(\hilbert{B}) $ be a linear map. The Choi operator $J({\mathcal{S}})\in \mathcal{L}(\mathcal{H}_B\otimes \mathcal{H}_A )$ of $\mathcal{S}$ is defined as 
\begin{align}
     J(\mathcal{S}) :=& \sum _{ij} \mathcal{S} (\ketbra{i}{j})\otimes \ketbra{i}{j} \\ 
      =&  (\mathcal{S} \otimes I_{A'})\left( \ketbra{\phi^+_d}\right)  d, 
\end{align}
where $\ket{\phi^+_d} := \frac{1}{\sqrt{d}} \sum_{i=0}^{d-1} \ket{i}\ket{i} $ is the maximally entangled state in dimension $d$. 
\end{definition}

In this work, we will primarily be interested in maps representing differences between two quantum channels, given by 
\begin{align}
    \Delta_{\Phi} := \Phi_1-\Phi_2, 
\end{align}
and their Choi operators $J(\Delta_{\Phi})$.

\subsection{Quantum channel discrimination}
The task of quantum channel discrimination consists in finding an optimal strategy to determine which one among a given, known set of channels, is realised~\cite{watrous_theory_2018}. Let $\{p_i,\Phi_i\}_{i=1}^{N}$ denote a set of $N$ channels $\Phi_i: \mathcal{L}(\hilbert{A}) \to\mathcal{L}(\hilbert{B})$, from which channel $\Phi_i$ is realised with probability $p_i$. In the case of channel discrimination, determining a strategy includes both selecting a suitable measurement and choosing an input state on which the unknown channel acts. The simplest type of strategy is to prepare an input state, let the unknown channel act on it, perform some measurement and then guess the channel based on the measurement outcome. This procedure is depicted in Figure~\ref{subfig:noent}. However, more generally, one can also prepare a bipartite quantum state $\rho_{AA'} \in \mathcal{L}(\hilbert{A} \otimes \hilbert{A'})$ and send one part of it through the channel before performing a measurement $M$ on the bipartite output state $(\Phi \otimes I_{A'})\rho_{AA'} \in \mathcal{L}(\hilbert{B} \otimes \hilbert{A'})$, as shown in Figure~\ref{subfig:anyent}. In the case when $\rho_{AA'}$ is a product state, this situation is equivalent to the simple situation in Figure~\ref{subfig:noent}, but if the bipartite input state is entangled, the correlations with the ancilla may make it possible to distinguish channels that are indistinguishable without entanglement, as shown in the introduction.
\begin{figure}[H]
    \centering
    \begin{subfigure}{0.32\textwidth}
        \centering
        \includegraphics[width=\linewidth]{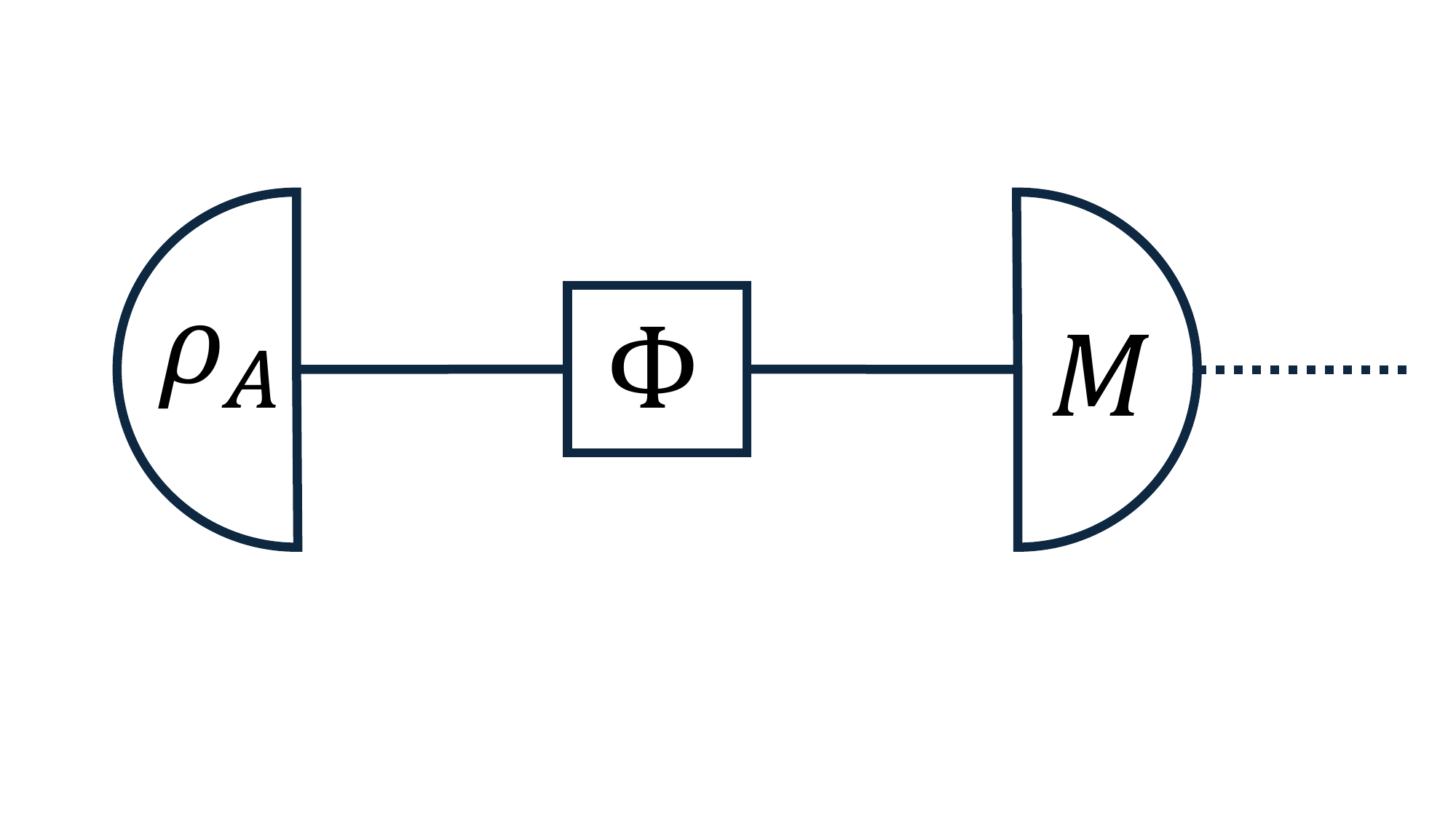}
        \caption{No entanglement.}
        \label{subfig:noent}
    \end{subfigure}
    \hfill
    \begin{subfigure}{0.32\textwidth}
        \centering
        \includegraphics[width=\linewidth]{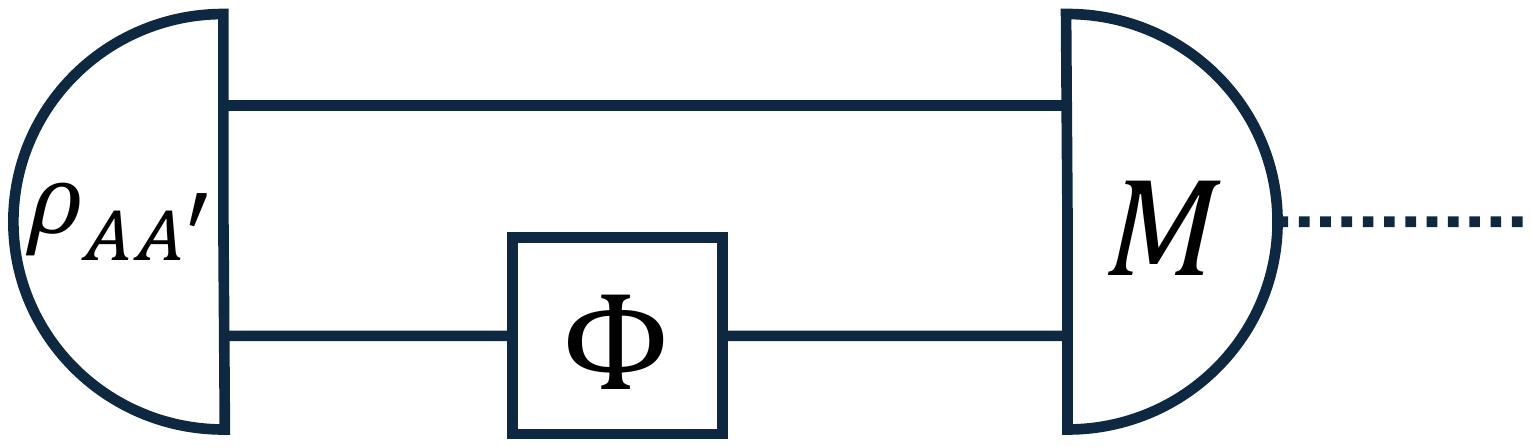}
        \caption{General bipartite input.}
        \label{subfig:anyent}
    \end{subfigure}
    \hfill
    \begin{subfigure}{0.32\textwidth}
        \centering
        \includegraphics[width=\linewidth]{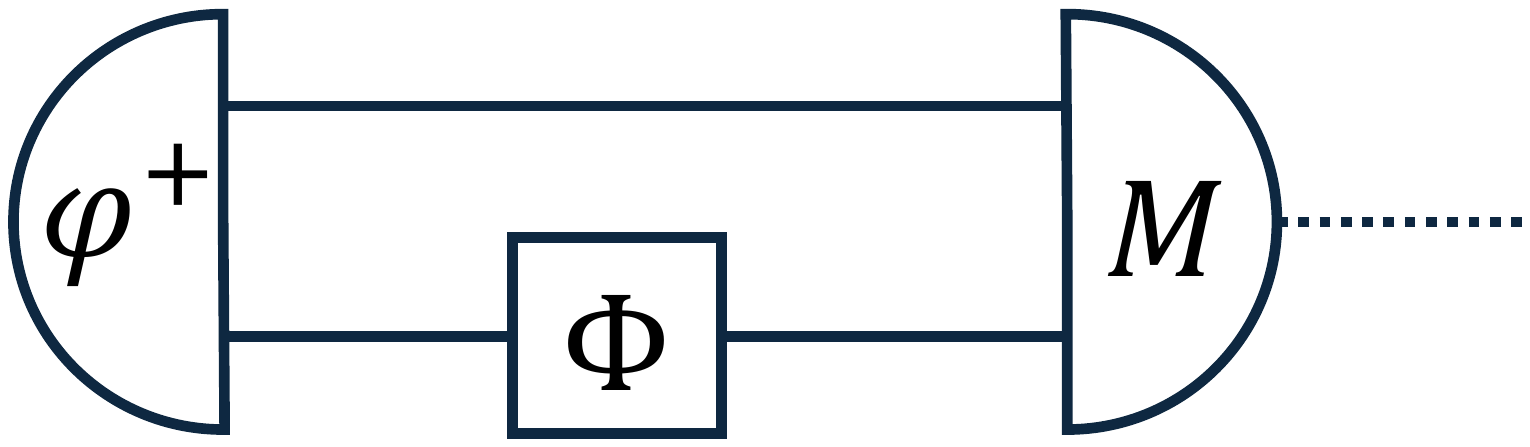}
        \caption{Max. entangled input.}
        \label{subfig:maxent}
    \end{subfigure}
    \caption[Strategies for measuring a quantum channel]{The figure shows three different types of strategies for measuring the channel $\Phi$. The simplest strategy, shown in (a), uses no ancilla. The most general strategy, shown in (b), uses a bipartite input state $\rho_{AA'}$ that may or may not be entangled with the auxiliary system $A'$. In (c), a maximally entangled input state $\phi^+$ is used, and we effectively measure the Choi state $ \frac{J(\Phi)}{d}$ of the channel $\Phi$.}
    \label{fig:ch_disc_strategies}
\end{figure}
The maximum probability of successfully distinguishing a pair of quantum channels $\Phi_1$ and $\Phi_2$ with uniform prior probabilities and $\Delta_{\Phi} = \Phi_1-\Phi_2$, is given by~\cite{watrous_theory_2018}
\begin{align}
    P_s^* = 
    \ \frac{1}{2} + \frac{1}{4} \norm{ \Delta_{\Phi}}_{\diamond}, 
\end{align}
where the diamond norm $\norm{\Delta_{\Phi}}_{\diamond}$ represents a distance between channels, and is defined as follows
\footnote{More generally, we can also define $ \Delta_{\Phi}(\lambda):= \lambda \Phi_1 - (1-\lambda)\Phi_2 $, where $\lambda\in [0,1]$ determines the prior probabilities for the two channels. In this case, the maximum success probability is $ P^* =\frac{1}{2} + \frac{1}{2} \norm{ \Delta_{\Phi}(\lambda)}_{\diamond}$. Although many results have straightforward generalisations to non-uniform prior probabilities (see e.g.~\cite{jencova_conditions_2016, puzzuoli_ancilla_2017}, in this work we restrict ourselves to the uniform case. } 
\begin{definition}[The diamond norm~\cite{watrous_theory_2018,kitaev_quantum_1997}]
    \label{def:diamondNorm}
   The diamond norm of a superoperator $\mathcal{S}: \mathcal{L}(\hilbert{A}) \to \mathcal{L}(\hilbert{B})$ is given by
    \begin{align}
    \norm{\mathcal{S}}_{\diamond} = \max_{\rho_{AA'}} \norm{(\mathcal{S} \otimes I_{A'}) \rho_{AA'}}_1,
    \end{align}
    where $\rho_{AA'}\in \mathcal{D}(\mathcal{H}_A\otimes\mathcal{H}_{A'})$ is a quantum state, i.e., $\rho_{AA'}\geq0$ and $\tr(\rho_{AA'})=1$.
\end{definition}
In particular, we will refer to the diamond norm of a difference map $\Delta_{\Phi} = \Phi_1-\Phi_2$, that is $\norm{\Delta_{\Phi}}_{\diamond}$, as the distinguishability between the two quantum channels $\Phi_1$ and $\Phi_2$ with a uniform probability distribution.

\subsection{Channel discrimination with a maximally entangled input state}
In the case when the bipartite input state is maximally entangled, as shown in Figure~\ref{subfig:maxent}, the channel discrimination task reduces to a state discrimination task between the Choi states of the channels. 

For state discrimination\footnote{For state discrimination, the maximum success probability of discrimination is given by $ P_s^* = \frac{1}{2} + \frac{1}{2} \norm{ \lambda \rho_0 - (1-\lambda) \rho_1}_1$ with prior probabilities determined by $\lambda \in [0,1]$ when we consider non-uniform probability distributions.}
, the maximum success probability is given by 
\begin{align}
    P_s^* = \frac{1}{2} + \frac{1}{4} \norm{ \rho_0 - \rho_1}_1, 
\end{align}
where $\rho_0, \rho_1$ are the states to be distinguished and $\norm{*}_1$ is the trace norm defined as $\norm{A}_1 =\operatorname{Tr}(\sqrt{A^{\dagger}A})$.

Since discrimination of Choi states of channels corresponds to channel discrimination with a maximally entangled input state, we define the following norm, which thus represents the distance between two channels when using a maximally entangled input state: 
\begin{definition}[The maximally entangled (ME) norm ]
    The ME-norm of a superoperator $\mathcal{S} : \mathcal{L}(\hilbert{A}) \to\mathcal{L}(\hilbert{B} )$ with an (unnormalised) Choi operator $J(\mathcal{S})$, is given by
    \begin{align}
        \norm{\mathcal{S}}_{\text{ME}} := \norm{\frac{J(\mathcal{S})}{d_A}}_1. 
    \end{align}
\end{definition}
The fact that the ME-norm is indeed a norm follows directly from the fact that the trace norm is a norm. This norm is also defined in~\cite{rosenthal_quantum_2024}, but with a different name (ACID-norm) and in a different context.  

In particular, the ME-norm of a difference map $\Delta_{\Phi} = \Phi_1-\Phi_2$ is given by 
      \begin{align}
        \norm{\Delta_{\Phi}}_{\text{ME}} = \norm{\frac{J(\Delta_{\Phi})}{d_A}}_1
    \end{align}
and will be called the ME-distinguishability between the two channels. 

\section{Entanglement and other general properties of the optimal input states for channel discrimination}
\label{ch:results}

In this section, we are interested in understanding what characterises the optimal input states for a specific discrimination task. In order to do this, we introduce the following quantity, which will prove useful throughout the rest of this work. 

\begin{definition}[M-operator]
\label{def:M_matrix}
 Let $\mathcal{S}: \mathcal{L}(\hilbert{A})\to \mathcal{L}(\hilbert{B})$ be a superoperator and $J(\mathcal{S})$ its Choi operator. Then the M-operator $M(\mathcal{S})$ is given by
\begin{align}
    \label{eq:M_matrix}
    M(\mathcal{S}) :=& \Tr_B\left[\sqrt{J(\mathcal{S})^\dagger J(\mathcal{S}}\,\right] \\
    =& \Tr_B\left[ \,\abs{J(\mathcal{S})}\,\right]. 
\end{align}
\end{definition}

In particular, for the difference map $\Delta_{\Phi} =\Phi_1-\Phi_2 $ between two quantum channels  $\Phi_1$, $\Phi_2: \mathcal{L}(\hilbert{A})\to \mathcal{L}(\hilbert{B})$, the M-operator 
\begin{align}
    M(\Delta_{\Phi}) :=& \Tr_B\left[\sqrt{J(\Delta_{\Phi})^\dagger J(\Delta_{\Phi})}\,\right] \\
    =& \Tr_B\left[ \,\abs{J(\Delta_{\Phi})}\,\right], 
\end{align}
will play an import role in characterising the optimal input state for different discrimination tasks. 

\subsection{Maximal entanglement best case pairs and maximal entanglement worst case pairs} 
In this section, we will study the extremal cases where a maximally entangled input state is either optimal or as suboptimal as it can be. Since the diamond norm and the ME-norm of the difference between two channels admit interpretations as the distance between the two with an optimised or maximally entangled input state, respectively, the two strategies can be compared by studying the relation between the two norms. The following result, introduced in e.g.~\cite{jencova_conditions_2016, brandao_generic_2015, watrous_theory_2018}, provides such a relation. 
\begin{proposition}[~\cite{jencova_conditions_2016, brandao_generic_2015, watrous_theory_2018}]
    \label{thm:inequality}
    Let $\mathcal{S}$ be a linear superoperator $\mathcal{S}: \mathcal{L}(\mathcal{H}_A) \rightarrow \mathcal{L}(\mathcal{H}_B)$ with $d = \dim({\mathcal{H}_A})$. Then,
    \begin{equation}
    \label{eq:norm_ineq_general}
    \norm{\mathcal{S}}_{\text{ME}} \leq \norm{\mathcal{S}}_{\diamond} \leq d\,\norm{\mathcal{S}}_{\text{ME}}, 
    \end{equation}
    where $\norm{\cdot}_{\text{ME}}$ represents the ME-norm and $\norm{\cdot}_{\diamond}$ represents the diamond norm. \newline
\end{proposition}
In particular, when $\mathcal{S} = \Delta_{\Phi} = \Phi_1-\Phi_2$ represents a difference between two CPTP maps $\Phi_1,\Phi_2$, we have
\begin{equation}
\label{eq:norm_ineq}
\norm{\Delta_{\Phi}}_{\text{ME}} \leq \norm{\Delta_{\Phi}}_{\diamond} \leq d\,\norm{\Delta_{\Phi}}_{\text{ME}}, 
\end{equation}
showing that the distinguishability between two quantum channels is upper bounded by $d$ times the corresponding ME-distinguishability. 
    
Later, in Corollary~\ref{cor:tigher_lower_bound} we make use of the M-operator (see Definition~\ref{def:M_matrix}) to prove a tighter lower bound for the diamond norm as a function of the corresponding ME-norm. Furthermore, a tighter upper bound is provided in~\cite{jencova_conditions_2016}), and for the special case of unitaries a tight upper bound is proven in~\cite{yamazaki_cliffordv_2025}. As we are interested in the interpretation of the upper bound as related to ME-distinguishability, we will mainly be concerned with the bound presented above. 

Since the norms are directly related to the optimal success probabilities of discrimination, $P_{ME}^* $ and $ P_\diamond^*$, the result above also provides the following inequalities for success probability of discrimination:
\begin{align}
        \label{eq:proba_ineq}
        P_{ME}^* \leq P_\diamond^* \leq \frac{1}{2} + d\,\left( P_{ME}^* - \frac{1}{2} \right)
    \end{align}
    or, equivalently,
    \begin{align}
    \left(P_{ME}^* - \frac{1}{2}\right) \leq \left(P_\diamond^* - \frac{1}{2}\right) \leq  d\,\left( P_{ME}^* - \frac{1}{2} \right).
\end{align}
This shows that there is an upper bound to how much better an optimised input state can perform compared to the maximally entangled one when discriminating between pairs of channels. Using a maximally entangled input state is always less then a factor $d$ times worse than using the optimal one, something which supports our understanding of the maximally entangled input state as a good input state when we have no information about the channels we want to distinguish. 

As we are interested in characterising the cases when maximal entanglement is optimal or maximally suboptimal, respectively, we introduce the terms Maximal Entanglement Best Case (MEBC) pairs and Maximal Entanglement Worst Case (MEWC) pair to mean pairs of channel where either of the bounds in Inequality~\ref{eq:norm_ineq} are saturated. Saturation of the lower bound implies by definition that a maximally entangled input state is optimal for discrimination of $\Phi_1$ and $\Phi_2$:  
\begin{definition}[Maximal Entanglement Best Case (MEBC) pair]
    \label{def:MEBC}
    Let $\Phi_1$, $\Phi_2: \mathcal{L}(\hilbert{A})\to \mathcal{L}(\hilbert{B})$ be two quantum channels with difference map $\Delta_{\Phi} =\Phi_1-\Phi_2 $. A maximally entangled input state is optimal for discrimination of $\Phi_1$ and $\Phi_2$  if and only if the lower bound in Inequality~\ref{eq:norm_ineq} is saturated, e.g. if and only if
    \begin{align}
        \norm{\Delta_{\Phi}}_{\text{ME}} = \norm{\Delta_{\Phi}}_{\diamond}.
    \end{align}
    Whenever this is the case, $\Phi_1$ and $\Phi_2$ will be called a Maximal Entanglement Best Case (MEBC) pair. 
\end{definition}

Similarly, we define Maximal Entanglement Worst Case pairs as follows:
\begin{definition}[Maximal Entanglement Worst Case (MEWC) pair]
    \label{def:MEWC}
    Let $\Phi_1$, $\Phi_2: \mathcal{L}(\hilbert{A})\to \mathcal{L}(\hilbert{B})$ be two quantum channels with difference map $\Delta_{\Phi} =\Phi_1-\Phi_2 $. A maximally entangled input state is maximally suboptimal for discrimination of $\Phi_1$ and $\Phi_2$  if and only if the upper bound in Inequality~\ref{eq:norm_ineq} is saturated, e.g if and only if
    \begin{align}
      \norm{\Delta_{\Phi}}_{\diamond} = d\,\norm{\Delta_{\Phi}}_{\text{ME}}.
    \end{align}
     Whenever this is the case, $\Phi_1$ and $\Phi_2$ will be called a Maximal Entanglement Worst Case (MEWC) pair, and the discriminability is a factor $d$ times worse than with the optimal strategy. 
\end{definition}

MEWC pairs are pairs of quantum channels for which the maximally entangled input state performs maximally bad as an input state. This does not mean that it is the worst possible input state for the discrimination task, but rather that there exists no other discrimination task where the maximally entangled input state performs worse relatively to the optimal strategy. 

The following Proposition, presented in~\cite{michel_note_2018}, provides necessary and sufficient conditions for saturation of the upper and lower bounds in inequality~\ref{eq:norm_ineq}:  

\begin{proposition}[~\cite{michel_note_2018}, theorems 1 and 2] 
    \label{thm:diamond_norm_bound_saturation}
    Let $\mathcal{S}$ be a superoperator $\mathcal{S}: \mathcal{L}(\mathcal{H}_A) \rightarrow \mathcal{L}(\mathcal{H}_B)$ with $d = \dim({\mathcal{H}_A})$ and $J(\Delta_{\Phi})$ its Choi operator. Then we have the following conditions for saturation of the upper and lower bounds of the diamond norm given in inequality~\ref{eq:norm_ineq}; $\norm{\mathcal{S}}_{\text{ME}} \leq \norm{\mathcal{S}}_{\diamond} \leq d\,\norm{\mathcal{S}}_{\text{ME}}$. \\ 
    \newline
    Lower bound saturation:
    \begin{align}
    \label{eq:cond_lower}
    \norm{\mathcal{S}}_{ME} = \norm{\mathcal{S}}_{\diamond}
    \iff
     M(\mathcal{S})
    = \norm{\mathcal{S}}_{ME}\, I_d. \textcolor{white}{mmm.}
    \end{align}
    Upper bound saturation:
    \begin{align}
    \label{eq:cond_upper}
    \norm{\mathcal{S}}_{\diamond} = d \norm{\mathcal{S}}_{ME}
    \iff
     & \exists\ \text{ a rank-1 projector } \\ 
     & \ketbra{\phi} \in \mathcal{L}(\mathcal{H}_A) \nonumber \\
    & \text{such that} \nonumber \\
     & M(\mathcal{S}) = \norm{J(\mathcal{S})}_1\, \ketbra{\phi}. \nonumber
    \end{align}
\end{proposition}

These conditions are proven in~\cite{michel_note_2018}, but with a slightly different notation and in a different context. They use a norm called the square norm, related to the diamond norm via
\begin{align}
    \norm{J(\mathcal{S})}_{\square} = d\norm{\mathcal{S}}_{\diamond}, 
\end{align}
where $d$ is the dimension of the input space.

We will now discuss the meaning of these conditions in the context of channel discrimination between two quantum channels $\Phi_1$ and $\Phi_2$ by means of the difference map $\Delta_{\Phi} = \Phi_1-\Phi_2$. Notice first that these condition precisely charecterises MEWC and MEBC pairs. Following Proposition~\ref{thm:diamond_norm_bound_saturation}, a $\Phi_1$ and $\Phi_2$ is a MEBC pair whenever $M(\Delta_{\Phi}) = \norm{\Delta_{\Phi}}_{\text{ME}}\, I_d$, in other words, when $M(\Delta_{\Phi})$ is full rank and all its eigenvalues are equal. A similar condition, there called the MEI condition, is presented in~\cite{jencova_conditions_2016}, stating that a maximally entangled input state is optimal if and only if $M(\Delta_{\Phi})$ is proportional to the identity operator. Based on this condition, they derive specific conditions for optimality of the maximally entangled input state for covariant channels, qubit channels, unitary channels and simple projective measurements.

On the other hand, $\Phi_1$ and $\Phi_2$ is a MEBC pair whenever  $M(\Delta_{\Phi})$ is rank $1$. In the next section~\ref{ch:properties}, we will show that saturation of the upper bound corresponds to cases when entanglement is of no use for channel discrimination and a separable input state is optimal.

\subsection{Properties of optimal input states for channel discrimination}
\label{ch:properties}
In the following, we will show that for MEWC pairs, a separable input state is optimal, and any entangled input state performs strictly worse. 

\begin{theorem}
    \label{thm:upper_bound_product}
    If $\Phi_1$, $\Phi_2: \mathcal{L}(\hilbert{A})\to \mathcal{L}(\hilbert{B})$ are a pair of Maximally Entangled Worst Case (Def.~\ref{def:MEWC}), then the optimal state for discriminating this pair is necessarily separable, and any entangled input state attains a strictly worse performance.

    Moreover, the M-operator $M(\Delta_{\Phi}) = \Tr_B\left[ \,\abs{J(\Delta_{\Phi})}\,\right]$ for the difference map $\Delta_{\Phi} =\Phi_1-\Phi_2 $ respects  $M(\mathcal{S}) = \norm{J(\mathcal{S})}_1\, \ketbra{\phi}$ for some pure quantum state $\ketbra{\phi}\in \mathcal{L}(\hilbert{A})$, and any optimal input states are necessarily given by
    \begin{align}
        \rho_{AA'} = \ketbra{\phi}\otimes \tau \in \mathcal{D}(\hilbert{A} \otimes \hilbert{A'})
    \end{align} 
    where $\tau \in \mathcal{D}(\hilbert{A'})$ is an arbitrary quantum state in the auxiliary system.
\end{theorem}

\begin{proof}
    Assume $M(\Delta_{\Phi}) = \Tr_B{\abs{J(\Delta_{\Phi})}} = \alpha \ketbra{\phi} $. 
    For all $\ket{\psi}$ orthogonal to $\ket{\phi}$, we have
    \begin{align}
        \bra{\psi} \Tr_B{\abs{J(\Delta_{\Phi})}} \ket{\psi} = 0 = \Tr\left((I_B \otimes \ketbra{\psi}) \abs{J(\Delta_{\Phi})}\right).
    \end{align}
    Since both $(I_B \otimes \ketbra{\psi})$ and $\abs{J(\Delta_{\Phi})}$ are positive operators, and the trace of their product is zero, we have
    \begin{align}
        (I_B \otimes \ketbra{\psi}) \abs{J(\Delta_{\Phi})} = 0, 
    \end{align}
    giving $\abs{J(\Delta_{\Phi})} = 0$ for all components in $\mathcal{H}_A$ orthogonal to $\ket{\phi}$.
    Thus we have $\text{supp}(\abs{J(\Delta_{\Phi})}) = \operatorname{supp}(J(\Delta_{\Phi}))  \subseteq  \mathcal{H}_B\, \otimes \,  \text{span}\{\ket{\phi}\}$
    and $J(\Delta_{\Phi}) \propto C \otimes \ketbra{\phi} $ for some hermitian operator $C \in \mathcal{L}(\mathcal{H}_{B}$). 
    To see that a product input state is optimal in this case, we use the following form of the diamond norm:
    \begin{align}
        \label{eq:variat}
        \norm{\Phi}_{\diamond} & = \max_{\sigma} \norm{(\sigma^{1/2} \otimes I)\,J(\Delta_{\Phi})(\sigma^{1/2} \otimes I)}_1 \\
        & = \max_{\sigma} \norm{(\sigma^{1/2} \otimes I)\,(\ketbra{\phi} \otimes C)(\sigma^{1/2} \otimes I)}_1\\
        & = \max_{\sigma}\bra{\phi}\sigma\ket{\phi}\ \norm{C}_1
    \end{align}
    where we have used the multiplicativity of the trace norm of a tensor product ($\norm{X \otimes Y}_1 = \norm{X}_1 \, \norm{Y}_1$) and the fact that $\norm{\sigma^{\frac{1}{2}} \ketbra{\phi} \sigma^{\frac{1}{2}}}_1 = \bra{\phi}\sigma\ket{\phi}$. 
    This is clearly maximised when $\sigma = \ketbra{\phi}$. Furthermore, since $\sigma = \Tr_{A'}(\rho_{AA'})$, the partial trace of the total input state $\rho_{AA'} \in \mathcal{L}(\mathcal{H}_A\otimes\mathcal{H}_{A'})$ is a pure state and the total input state $\rho_{AA'}$ must be a product state. Finally, if the total input state $\rho_{AA'}$ is entangled, $\sigma = \Tr_{A'}(\rho_{AA'})$ is a mixed state and $\bra{\phi}\sigma\ket{\phi}$ is strictly smaller than one. 
\end{proof}
This shows that MEWC pairs correspond to discrimination tasks where a separable input state is optimal and any entanglement, in fact, reduces the discriminability between the two channels. This is the case for the example in Figure~\ref{fig:m_plot}. 

Theorem~\ref{thm:upper_bound_product} presents a strong constraint on optimal input states for discriminating a pair of channels which are MEWC based on their associated M-operator. 
In the following, we generalise this result to identify a necessary condition for optimal input states in \textit{any} discrimination task between a pair of channels $\Phi_1$ and $\Phi_2$. 

\begin{theorem}
    \label{thm:rankM_thm}
    If $\rho_{AA'}\in \mathcal{D}(\hilbert{_A}\otimes\hilbert{A'})$ is an optimal state for discriminating the pair of channels $\Phi_1,\Phi_2 : \mathcal{L}(\mathcal{H}_A) \rightarrow \mathcal{L}(\mathcal{H}_B)$, then
     \begin{align}
         \Tr_{A'}(\rho_{AA'})\in \mathcal{D}\big(\operatorname{supp}(M(\Delta_{\Phi}))\big),
     \end{align}
      where $M(\Delta_{\Phi}) = \Tr_B\left[ \,\abs{J(\Delta_{\Phi})} \, \right]$ is the M-operator as in Definition~\ref{def:M_matrix}. 
\end{theorem}

\begin{proof}
Let us define $\mathcal{S} := \operatorname{supp}(M(\Delta_{\Phi}))$.
    Similarly as in the proof of Theorem~\ref{thm:upper_bound_product}, we can show that for every $\ket{\psi} \in \mathcal{S}^{\perp}$,  
    \begin{align}
        (I_B \otimes \ketbra{\psi}) \abs{J(\Delta_{\Phi})} = 0, 
    \end{align}
    giving $\abs{J(\Delta_{\Phi})} = 0$ for all components in $\mathcal{H}_A$ orthogonal $\mathcal{S}$. Therefore, $\operatorname{supp}(J(\Delta_{\Phi}))  \subseteq  \mathcal{H}_B \otimes \, \mathcal{S}$. 
    
    Let $\sigma$ be an arbitrary density matrix and 
    \begin{align}
    \sigma' = \frac{P_{\mathcal{S}}\sigma P_{\mathcal{S}}}{\Tr(P_{\mathcal{S}}\sigma P_{\mathcal{S}})}
    \end{align}
    its normalised compression onto $\mathcal{S}$, where $P_{\mathcal{S}}$ is a projector onto the space $\mathcal{S}$. Inserting this expression in the variational form of the diamond norm (~\ref{eq:variat}) gives 
    \begin{align}
    & \norm{(\sigma'^{\frac{1}{2}} \otimes I)\,J(\Delta_{\Phi})(\sigma'^{\frac{1}{2}} \otimes I)}_1 \\
    & = \frac{1}{\Tr(P_{\mathcal{S}}\sigma P_{\mathcal{S}})}\norm{([P_{\mathcal{S}}\sigma P_{\mathcal{S}}]^\frac{1}{2} \otimes I)\,J(\Delta_{\Phi})([P_{\mathcal{S}}\sigma P_{\mathcal{S}}]^\frac{1}{2} \otimes I)}_1 \nonumber \\
    & = \frac{1}{\Tr(P_{\mathcal{S}}\sigma P_{\mathcal{S}})}\norm{(\sigma^\frac{1}{2} \otimes I)\,J(\Delta_{\Phi})(\sigma^\frac{1}{2} \otimes I)}_1, \nonumber 
    \end{align}
    where the final inequality arises from the fact that $J(\Delta_{\Phi})$ only has support on $ \mathcal{H}_B \, \otimes \, \mathcal{S} $, so only the part of $\sigma$ that is in $\mathcal{S}$ can give a contribution to the norm. $\Tr(P_{\mathcal{S}}\sigma P_{\mathcal{S}})$ is strictly smaller than one as long as $\sigma$ is not in $\mathcal{S}$, so for any input state $\sigma$, the norm 
    \begin{align}
        \norm{(\sigma^{1/2} \otimes I)\,J(\Delta_{\Phi})(\sigma^{1/2} \otimes I)}_1 
    \end{align}
    will increase by instead using its compression onto $\mathcal{S}$. We can therefore conclude that the optimal input state $\Tr_{A'}\ketbra{\Psi}$ must be a density operator on $\mathcal{S}$. Furthermore, if an input state has rank higher than $\rank{M(\Delta_{\Phi})}$, it must necessarily have components outside of $\mathcal{S}$, and thus it cannot be optimal, since the discriminability would increase by instead using its compression onto $\mathcal{S}$. 
\end{proof}
Following this theorem, we note that in any discrimination task, if an input state has Schmidt rank higher than the rank of the M-operator associated to the task, it cannot be optimal.

This theorem shows that there is, in some sense, a limit to how much entanglement, quantified by the Schmidt rank of the input state, that is beneficial in the discrimination task of two specific, known quantum channels. Entanglement is only useful within the subspace on which the map describing the difference between the two quantum channels works non-trivially. $M(\Delta_{\Phi})$ identifies exactly this subspace. Whereas maximal entanglement is often regarded as optimal for channel discrimination in a universal or average sense,~\cite{caiaffa_channel_2018, bae_more_2019}, it is evident that for specific discrimination tasks, this is generally not the case. Optimality of the maximally entangled input state corresponds, in fact, to the very specific cases when the difference map $\Delta_{\Phi}$ works non-trivially and equally on all directions in the eigenspace of $M(\Delta_{\Phi})$. Intuitively, and as shown in this section, channels for which the difference map $\Delta_{\Phi}$ has an ``asymmetrical'' action and/or works only on a subspace of the input space, are better distinguished by tailored input states with reduced entanglement. In the extreme MEWC case where the difference map only ``sees'' one specific direction, an unentangled input state in this direction is optimal.

\subsection{A tighter inequality for the diamond norm and the ME norm}
The quantity $\Tr(P_{\mathcal{S}}\sigma P_{\mathcal{S}}) \in (0,1)$ as defined in the proof of Theorem~\ref{thm:rankM_thm} does, in some sense, provide a quantitative way of describing the disadventageousness of a state $\sigma$, as we have the following
 \begin{align}
    \norm{(\sigma^{1/2} \otimes I)\,J(\Delta_{\Phi})(\sigma^{1/2} \otimes I)}_1 \\
    =  \Tr(P_{\mathcal{S}}\sigma P_{\mathcal{S}}) \norm{(\sigma'^{1/2} \otimes I)\,J(\Delta_{\Phi})(\sigma'^{1/2} \otimes I)}_1 \\
    \leq \Tr(P_{\mathcal{S}}\sigma P_{\mathcal{S}}) \norm{\Delta_{\Phi}}_{\diamond}, 
\end{align}
where all quantities are defined as in the proof of Theorem~\ref{thm:rankM_thm}. This shows that the maximum possible discriminability with an input state $\sigma$ is given by the factor $\Tr(P_{\mathcal{S}}\sigma P_{\mathcal{S}})$ times the optimal discriminability. 

Based on this, we can show the following result for the optimality of the maximally entangled state when we know the rank of $M(\Delta_{\Phi})$: 
\begin{corollary}
\label{cor:tigher_lower_bound}
    Let $\Phi_1,\Phi_2 : \mathcal{L}(\mathcal{H}_A) \rightarrow \mathcal{L}(\mathcal{H}_B)$ be two quantum channels and $\Delta_{\Phi}= \Phi_1-\Phi_2$ their difference. Let $J(\Delta_{\Phi})$ be the Choi operator of $\Delta_{\Phi}$ and $M(\Delta_{\Phi}) = \Tr_B\left[ \,\abs{J(\Delta_{\Phi})} \, \right]$ with rank r, as defined as in~\ref{eq:M_matrix}. Then we have
    \begin{align}
    \norm{\Delta_{\Phi}}_{\text{ME}} \leq \frac{r}{d}\norm{\Delta_{\Phi}}_{\diamond}
     \end{align}
\end{corollary}
\begin{proof}
    From the discussion above, we have 
     \begin{align}
    &\norm{(\sigma^{1/2} \otimes I)\,J(\Delta_{\Phi})(\sigma^{1/2} \otimes I)}_1 \\
    &\leq \Tr(P_{\mathcal{S}}\sigma P_{\mathcal{S}}) \norm{\Delta_{\Phi}}_{\diamond}, 
    \end{align}
    where $\sigma$ is an arbitrary reduced input state and $P_{\mathcal{S}}$ is a projector onto the support of $M(\Delta_{\Phi})$, denoted $\mathcal{S}$, and $r=\operatorname{rank}(M(\Delta_{\Phi}))$. If we let $\sigma = \frac{I}{d}$ be the maximally mixed state, meaning that the total input state is maximally entangled, we have
    \begin{align}
        \Tr(P_{\mathcal{S}}\sigma P_{\mathcal{S}}) = \Tr(P_{\mathcal{S}}\frac{I}{d} P_{\mathcal{S}}) = \frac{1}{d}\Tr( P_{\mathcal{S}}) = \frac{r}{d}. 
    \end{align}
    Moreover, we recognise that 
    \begin{align}
        &\norm{(\sigma^{1/2} \otimes I)\,J(\Delta_{\Phi})(\sigma^{1/2} \otimes I)}_1 \\
        =& \frac{1}{d}\norm{(I \otimes I)\,J(\Delta_{\Phi})(I \otimes I)}_1 \\
        = &\norm{\frac{J(\Delta_{\Phi})}{d}}_1 = \norm{\Delta_{\Phi}}_{\text{ME}}, 
    \end{align}
    giving
    \begin{align}
        \norm{\Delta_{\Phi}}_{\text{ME}} \leq \frac{r}{d} \norm{\Delta_{\Phi}}_{\diamond}. 
    \end{align}
\end{proof}
This offers a tighter lower bound for the diamond norm: 
\begin{align}
\norm{\Delta_{\Phi}}_{\text{ME}} \leq \frac{d}{r}\norm{\Delta_{\Phi}}_{\text{ME}} \leq \norm{\Delta_{\Phi}}_{\diamond} \leq d\,\norm{\Delta_{\Phi}}_{\text{ME}}, 
\end{align}
Again, we see that for $r=1$, the lower and upper bounds for $\norm{\Delta_{\Phi}}_{\diamond} $ coincide, and we have equality in the upper bound. If $r=d$, the lower bound coincides with the original lower bound, allowing for (but not forcing) ME optimality.

Throughout this section, we see that the M-operator reveals useful information about an arbitrary channel pair discrimination task. It tells us whether a maximally entangled input state is optimal, whether we need entanglement at all, and also what is the maximal useful ancillary dimension for a discrimination task. These properties can therefore be found without the need for any optimisation.  

\section{Finding channels that are Maximal Entanglement Worst Case pairs}
\subsection{MEWC pairs of measurement channels}
\label{ch:measurement_ch}
One class of channels for which we have found examples of MEWC pairs are the measurement channels mentioned in the introduction, defined as follows: 
\begin{definition}[Measurement channel] 
\label{def:measurement_channel}
    Let $M = \{M_i\}_{i=0}^{n-1}$ be a POVM with $n$ outcomes in dimension $d$, and let $\rho \in \mathcal{D}(\hilbert{A})$ be a density operator of dimension $d$. The measurement channel corresponding to the POVM $M$ is given by
    \begin{align}
        \Phi_M(\rho)  = \sum_{i=0}^{n-1} \Tr(\rho M_i)\, \ketbra{i}.  
    \end{align}
\end{definition}
The following theorem provides conditions for when a pair of measurement channels is MEWC. 
\begin{theorem}
    \label{thm:meas_cond}
    Let $\Phi_M,\Phi_N: \mathcal{L}(\hilbert{A}) \to \mathcal{L}(\hilbert{B})$ be two measurement channels corresponding to the two $n$-outcome POVMs $M =\{M_i\}_{i=0}^{n-1}$ and $N = \{N_i\}_{i=0}^{n-1}$. The channels  $\Phi_M, \Phi_N$ form a Maximal Entanglement Worst Case pair if and only if there exists a pure state $\ketbra{\phi}\in\mathcal{L}(\hilbert{A})$ and coefficients $\gamma_i\in\mathbb{C}$ such that 
        \begin{align}
         (M_i-N_i) = \gamma_i \ketbra{\phi}, \quad \forall i \in \{0,\ldots , n-1\}.
    \end{align}

    Moreover, the optimal input state is necessarily given by $\rho_{AA'}=\ketbra{\phi} \otimes \tau$, where $\tau \in \mathcal{D}(\hilbert{A'})$ is an arbitrary quantum state.
\end{theorem}
The proof of this result will make use of the following lemma:
\begin{lemma}
    \label{lm:m_lemma}
    Let $M(\Delta_{\Phi}) = \Tr_B\left[ \,\abs{J(\Delta_{\Phi})} \, \right]$ for a pair of channels in dimension $d$ as before. Then we have
    \begin{align}
        M(\Delta_{\Phi}) \propto \ketbra{\phi} \iff M({\Delta_{\Phi}}) = d \, \norm{\Delta_{\Phi}}_{\text{ME}} \ketbra{\phi}
    \end{align}
    and 
     \begin{align}
        M(\Delta_{\Phi}) \propto I_d \iff M({\Delta_{\Phi}}) = \norm{\Delta_{\Phi}}_{\text{ME}} \, I_d. 
    \end{align}
    
\end{lemma}
\begin{proof}[Proof of Lemma \ref{lm:m_lemma}]
    We have 
    \begin{align}
        &\Tr(M(\Delta_{\Phi})) = \Tr(\Tr_B\!\left[ \abs{J(\Delta_{\Phi})}\right]) \\
        =&\Tr\left( \sqrt{J(\Delta_{\Phi})^{\dagger}J(\Delta_{\Phi})} \right) \\
        =&\norm{J(\Delta_{\Phi})}_1= d \,\norm{\Delta_{\Phi}}_{\text{ME}}. 
    \end{align}
    
    Firstly, let $M(\Delta_{\Phi}) = \alpha\ketbra{\phi}$. Then,
    \begin{align}
        \Tr(M(\Delta_{\Phi})) = \alpha, 
    \end{align}
    so necessarily, $\alpha = d \,\norm{\Delta_{\Phi}}_{\text{ME}}$. 

    Similarly, if $M(\Delta_{\Phi}) = \beta I_d$, then
        \begin{align}
        \Tr(M(\Delta_{\Phi})) = \beta \, d, 
    \end{align}
    so necessarily, $\beta = \norm{\Delta_{\Phi}}_{\text{ME}}$. 
\end{proof}
\begin{proof}[Proof of Theorem \ref{thm:meas_cond}]
    Let $\Delta_{\Phi} = \Phi_M-\Phi_N$. We begin by showing that $M(\Delta_{\Phi}) = \sum_i \abs{(M_i-N_i)}$ for measurement channels. 
        The Choi operator of a measurement channel is given by 
    \begin{align}
        &J(\Phi_M) = \sum_{ab}\Phi_1(\ketbra{a}{b})\otimes\ketbra{a}{b} \\ 
        =& \sum_{ab}\left[\sum_i \Tr\left(M_i \ketbra{a}{b}\right) \ketbra{i} \right]\otimes\ketbra{a}{b} \\
        =& \sum_i \sum_{ab} \bra{b}M_i\ket{a} \ketbra{i} \otimes\ketbra{a}{b}\\ 
        =& \sum_i   \ketbra{i} \otimes M_i^T
    \end{align}
    where we have used that the transpose of $M_i$, $M_i^T = \sum_{ab} \bra{b}M_i\ket{a} \ketbra{a}{b}$. 
    Moreover, the Choi operator of $\Delta_{\Phi} = \Phi_M-\Phi_N$ is given by
    \begin{align}
        J(\Delta_{\Phi}) = \sum_i   \ketbra{i} \otimes (M_i^T-N_i^T) 
    \end{align}
    and 
    \begin{align}
        &M(\Delta_{\Phi}) = \Tr_B \left[\abs{J(\Delta_{\Phi})}\right]\\
        =& \Tr_B \left[\sum_i   \ketbra{i} \otimes \abs{(M_i-N_i)^T }\right] \\
        =& \Tr_B \left[\sum_i   \ketbra{i} \otimes \abs{(M_i-N_i) }\right]\\ 
        =& \sum_i \abs{(M_i-N_i)}. 
    \end{align}

    From Proposition \ref{thm:diamond_norm_bound_saturation} and Lemma \ref{lm:m_lemma}, we know that we have a MEWC pair if and only if $\sum_i \abs{(M_i-N_i)}$ is proportional to the rank-1 projector $\ketbra{\phi}$, and this is the case if and only if all the terms $\abs{(M_i-N_i)}$ are proportional to $\ketbra{\phi}$. Furthermore, if $\abs{(M_i-N_i)}$ is proportional to $\ketbra{\phi}$, then $(M_i-N_i)$ also has to be proportional to $\ketbra{\phi}$, which concludes the proof. 
\end{proof}
We now consider the special case of discrimination between two $2$-outcome qubit measurement channels. 
\begin{corollary}
    \label{corr:2_outcome_meas}
Let $\Phi_M,\Phi_N: \mathcal{L}(\hilbert{A}) \to \mathcal{L}(\hilbert{B})$ be a pair of measurement channels corresponding to the two $2$-outcome POVMs $M =\{M_i\}_{i=0}^{1}$ and $N = \{N_i\}_{i=0}^{1}$ respecting $M_1\neq N_1$ in dimension $d=2$. The pair of qubit channels $\Phi_M,\Phi_N$ is a Maximal Entanglement Worst Case pair if and only if 
\begin{align}
     \det(M_0-N_0) = 0 .
\end{align}
\end{corollary}

\begin{proof}
    Since $M_1 = I-M_0$ and $N_1 = I-N_0$ in the $2$-outcome case, we have 
    \begin{align}
        M(\Delta_{\Phi}) =& \sum_i \abs{(M_i-N_i)} \\
        =&\abs{M_0 -N_0} + \abs{(M_1-N_1)} \\
        =& \abs{M_0 -N_0} + \abs{(I-M_0-(I-N_0))} \\
        =& 2 \abs{M_0-N_0} 
    \end{align}
    Theorem~\ref{thm:diamond_norm_bound_saturation} says that MEWC pairs correspond to situations when $ M(\Delta_{\Phi})$ is a rank one projector. If $2 \abs{M_0-N_0} $ is rank one, then $M_0-N_0$ has to be rank one. In the qubit case, this means that the determinant has to be zero. 
\end{proof} 

Note that even though these channels may appear ``classical'' in the sense that they output a classical probability distribution, entanglement generally does increase the success probability of discriminating such channels. The MEWC scenario is only a special case of measurement channel discrimination, and there also exists measurement channel pairs that are MEBC, meaning that maximal entanglement is optimal. 

In the following, we will consider some examples of MEWC pairs for different types of measurement channels.
\begin{corollary} [Commuting measurement channel pairs that are MEWC]
    Let $\Phi_M,\Phi_N: \mathcal{L}(\hilbert{A}) \to \mathcal{L}(\hilbert{B})$ be a pair of measurement channels corresponding to the two commuting 2-outcome POVMs $M =\{M_i\}_{i=0}^{1}$ and $N = \{N_i\}_{i=0}^{1}$ with $[M_0,N_0] = [M_1,N_1] = 0$ . The pair of qubit channels $\Phi_M,\Phi_N$ is a Maximal Entanglement Worst Case pair if and only if $M_0$ and $N_0$ only have one unequal eigenvalue. For instance, the following commuting POVMs give rise to a MEWC pair of quantum channels: 
    \begin{align}
           M_0 = & 
        \begin{pmatrix}
            a & 0\\
            0 & b
        \end{pmatrix}, \qquad  M_1 = I-M_0
\end{align}
and
\begin{align}
           N_0 = & 
        \begin{pmatrix}
            a & 0\\
            0 & c
        \end{pmatrix}, \qquad  N_1 = I-N_0
\end{align}
such that 
\begin{align}
         M_0-N_0 = 
           \begin{pmatrix}
            0 & 0\\
            0 & b-c
        \end{pmatrix} ,
\end{align}   
where $a,b,c \in(0,1)$
\end{corollary}
\begin{proof}
    Using the same notation and definitions as above, we have
    \begin{align}
& M(\Delta_{\Phi}) = 2\abs{M_0-N_0}  \\  \nonumber
& =
        2\abs{\begin{pmatrix}
\mu_1 & 0 & \cdots & 0 \\
0 & \mu_2 & \cdots & 0 \\
\vdots & \vdots & \ddots & \vdots \\
0 & 0 & \cdots & \mu_d
\end{pmatrix}
-
\begin{pmatrix}
\nu_1 & 0 & \cdots & 0 \\
0 & \nu_2 & \cdots & 0 \\
\vdots & \vdots & \ddots & \vdots \\
0 & 0 & \cdots & \nu_d
\end{pmatrix}}\\ \nonumber
& = 
2\abs{
\begin{pmatrix}
\mu_1 - \nu_1 & 0 & \cdots & 0 \\
0 & \mu_2 - \nu_2 & \cdots & 0 \\
\vdots & \vdots & \ddots & \vdots \\
0 & 0 & \cdots & \mu_d - \nu_d 
\end{pmatrix}} \nonumber, 
    \end{align}
    where $\mu_1,\cdots,\mu_d$ and $\nu_1,\cdots,\nu_d$ are the eigenvalues of the POVM elements $M_0$ and $N_0$, respectively. We see that the rank of $M(\Delta_{\Phi})$ is given by the number of unequal eigenvalues between $M_0$ and $N_0$. According to Theorem~\ref{thm:rankM_thm}, this rank upper bounds the Schmidt rank of the optimal input state. When $M_0$ and $N_0$ have all but one eigenvalues in common, $M(\Delta_{\Phi})$ is a rank-1 projector, and $\Phi_M, \Phi_N$ is a MEWC pair.
    
\end{proof}
    These are special cases when the two measurements act completely equally on all but one direction in the input space (in the basis in which they both are diagonal), and intuitively, the best way to discriminate between them is to use an input state that only carries weight in the specific direction of the eigenvector corresponding to the unequal eigenvalue, that is, an unentangled state as described in Theorem~\ref{thm:upper_bound_product}. 
    
    Since the two measurements are diagonal in the same basis, one can argue that this is a discrimination between classical measurements, and that it is therefore intuitive that entanglement is not a resource for discrimination. However, in general, entanglement does increase discriminability between commuting measurement channels. There even exists examples of discrimination tasks between commuting measurement channels where the optimal input state is maximally entangled, as shown in the following result:

\begin{corollary}[Commuting measurement channel pairs that are MEBC]
    Let $\Phi_M,\Phi_N: \mathcal{L}(\hilbert{A}) \to \mathcal{L}(\hilbert{B})$ be a pair of measurement channels corresponding to the two commuting POVMs $M =\{M_i\}_{i=1}^{n}$ and $N = \{N_i\}_{i=1}^{n}$ with $[M_0,N_0] = [M_1,N_1] = 0$. There exists such pairs of channels that are Maximally Entangled Best Case pairs. 
    For instance the measurement channels corresponding to the following POVMs form a MEBC pair.
    \begin{align}
           M_0 = & 
        \begin{pmatrix}
            0.9 & 0\\
            0 & 0.6
        \end{pmatrix}, \qquad  M_1 = I-M_0
\end{align}
and
\begin{align}
           N_0 = & 
        \begin{pmatrix}
            0.5 & 0\\
            0 & 0.2
        \end{pmatrix}, \qquad  N_1 = I-N_0
\end{align}
such that 
\begin{align}
         M_0-N_0 = 
           \begin{pmatrix}
            0.4 & 0\\
            0 & 0.4
        \end{pmatrix} 
        = 0.4 I_2
\end{align}   
\end{corollary}
\begin{proof}
    Following directly from Proposition \ref{thm:diamond_norm_bound_saturation} and Lemma \ref{lm:m_lemma}, we have a MEBC pair whenever $M(\Delta_{\Phi})$ is proportional to identity. In the case of dichotomic measurement channels, we have $M(\Delta_{\Phi}) = 2\abs{M_0-N_0}$, as shown in the proof of Corollary \ref{corr:2_outcome_meas}, and since $M_0-N_0 = 0.4 I$ in this example, it is a MEBC pair. 
\end{proof}

Furthermore, we will study what the MEWC condition implies for projective measurement channels. 
\begin{corollary}[Projective measurement channel pairs that are MEWC]
        Let $\Phi_M,\Phi_N: \mathcal{L}(\hilbert{A}) \to \mathcal{L}(\hilbert{B})$ be a pair of measurement channels corresponding to the two projective measurements (PVMs) $M =\{M_i\}_{i=0}^{1}$ and $N = \{N_i\}_{i=0}^{1}$ such that $M_0 = M_0^2,M_1 = M_1^2, M_0M_1 = 0$ and similarly for the PVM $N$. Then, the only possibility of forming a MEWC pair is if $M$ and $N$ are different course-grainings of the same PVM. For instance, the following pair of measurements give rise to a MEWC pair of channels: 
        \begin{align}
        M_0 = 
        \begin{pmatrix}
            1 & 0 & 0\\
            0 & 1 & 0\\
            0& 0 & 0
        \end{pmatrix}, 
        \, M_1 = 
          \begin{pmatrix}
            0 & 0 & 0\\
            0 & 0 & 0\\
            0& 0 & 1
        \end{pmatrix}\\
            N_0 = 
        \begin{pmatrix}
            1 & 0 & 0\\
            0 & 0 & 0\\
            0& 0 & 0
        \end{pmatrix}, 
        \, N_1 = 
          \begin{pmatrix}
            0 & 0 & 0\\
            0 & 1 & 0\\
            0& 0 & 1
        \end{pmatrix}
    \end{align}
    such that 
    \begin{align}
        M_0-N_0 = 
          \begin{pmatrix}
            0 & 0 & 0\\
            0 & 1 & 0\\
            0& 0 & 0
        \end{pmatrix}. 
    \end{align}
    These measurements can be recognised as different coarse-grainings of the projective measurement
        \begin{align}
        \left\{
        \begin{pmatrix}
            1 & 0 & 0\\
            0 & 0 & 0\\
            0& 0 & 0
        \end{pmatrix}, 
        \,  
          \begin{pmatrix}
            0 & 0 & 0\\
            0 & 1 & 0\\
            0& 0 & 0
        \end{pmatrix},
        \, 
          \begin{pmatrix}
            0 & 0 & 0\\
            0 & 0 & 0\\
            0& 0 & 1
        \end{pmatrix} \right\}, 
        \end{align}
        Furthermore, in this case we can always achieve perfect discrimination. 
\end{corollary}
\begin{proof}
Assume $M_0-N_0 = c\ketbra{\phi}$ (from Corollary~\ref{corr:2_outcome_meas}) and let $M_0 = M_0^2 \, , N_0 = N_0^2$ since $M$ and $N$ are projective measurements. Then we have,
    \begin{align}
        N_0 & = M_0-c\ketbra{\phi} = N_0^2 \\
        &= (M_0-c\ketbra{\phi})^2, \\
        M_0-c\ketbra{\phi} 
         = M_0^2 &-cM_0\ketbra{\phi} - c\ketbra{\phi}M_0 \\
                 &+c^2\ketbra{\phi}, \\
     -\ketbra{\phi}& = -M_0\ketbra{\phi} - \ketbra{\phi}M_0 \\
                            &+c\ketbra{\phi},
    \end{align}
    \begin{align}
                 M_0\ket{\phi} = (1-\bra{\phi}M_0\ket{\phi}+c)\ket{\phi}, 
    \end{align}
    where we have used that $M_0 = M_0^2$ and applied $\ket{\phi}$ on both sides of the equation from the right. We also assume that $c \neq 0$. If we define $\bra{\phi}M_0\ket{\phi} := a$, and note that since $M_0$ is a projector it can only have eigenvalues $0$ or $1$, we get the following two possible cases:
    \begin{align}
        \text{Eigenvalue 1: } 1-a+c = 1-1+c = 1 \\ \implies M_0-N_0 = 1\ketbra{\phi}\\
        \text{Eigenvalue 0: } 1-a+c = 1-0+c = 0 \\ \implies M_0-N_0 = -1\ketbra{\phi}
    \end{align}
    So for projective $2$-outcome measurements, $M_0$ and $N_0$ must take the general form $M_0=P\, , N_0 = P \pm \ketbra{\phi}$, where $P$ is a projector and $M_0-N_0 = \pm \ketbra{\phi}$. In other words, the two projective measurements perform projections onto subspaces that are equal except for in one direction, given by $\ketbra{\phi}$. 
    
    We note that whenever this is the case, $M_0$ and $N_0$ necessarily commute: 
    \begin{align}
        [M_0,N_0] =& \left[P,P\pm \ketbra{\phi}\right]\\
        =&PP \pm P\ketbra{\phi} -PP \mp \ketbra{\phi} P \\
        =& 0, .  
    \end{align}
    where the two terms $P\ketbra{\phi}$ and $\ketbra{\phi} P$ cancel out due to the fact that $\ket{\phi}$ is an eigenvector of $P = M_0$. 
    
    Thus the only possible pairs of projective measurements are MEWC have $M_0$ and $N_0$ as almost equal diagonal matrices with zeros and ones as eigenvalues, with the only difference being that one of them has one extra one-value, which gives different course-grainings of some other PVM. Furthermore, using the input state corresponding to the differing eigenvalue unambiguously identifies the correct channel, so we have perfect discrimination.    
\end{proof}
 Note that for the qubit case, this can never happen non-trivially. One of the two channels necessarily becomes a trivial one with one identity element and one zero element. Consequently, two non-trivial projective dichotomic qubit measurements can never be a MEWC pair. This is, however, not the case for qubit measurement channels when we allow one of the to correspond to a general POVM, as shown in thefollowing example.

\begin{corollary}[MEWC pairs of qubit dichotmic measurements where one of them is projective]
Let $\Phi_M: \mathcal{L}(\hilbert{A}) \to \mathcal{L}(\hilbert{B})$ be a measurement channel corresponding to the projective qubit measurement  $M =\{M_i\}_{i=0}^{1}$. The only way to form a MEWC pair is if the other measurement channel  $N = \{N_i\}_{i=0}^{1}$ corresponds to a biased version of the same POVM. For instance, if $M$ is given by 
\begin{align}
         M_0 = &
        \begin{pmatrix}
            1 & 0\\
            0 & 0
        \end{pmatrix}, \, M_1 =
        \begin{pmatrix}
            0 & 0\\
            0 & 1
        \end{pmatrix},  
\end{align}
then $N$ is necessarily given by 
\begin{align}
         N_0 = &
        \begin{pmatrix}
            \alpha & 0\\
            0 & 0
        \end{pmatrix}, \, N_1 =
        \begin{pmatrix}
            1-\alpha & 0\\
            0 & 1
        \end{pmatrix}
\end{align}
with $\alpha \in [0,1]$. 
\end{corollary}
\begin{proof}
    We know from Corollary~\ref{corr:2_outcome_meas} that $\Phi_M$ and $\Phi_N$ form a MEWC pair whenever $\det(M_0-N_0)= 0$. Let 
    \begin{align}
        M_0-N_0 = &
        \begin{pmatrix}
            \lambda_0 & 0\\
            0 & \lambda_1
        \end{pmatrix}
        -
        \begin{pmatrix}
            \ \alpha & \beta\\
            \beta^* & \gamma
        \end{pmatrix}
        \\
        = & 
         \begin{pmatrix}
            \ \lambda_0-\alpha & \beta\\
            \beta^* & \lambda_1-\gamma
        \end{pmatrix}.
    \end{align}
    This is a rank-1 projector whenever
    \begin{align}
    \label{eq:qubit_2_out_eq}
        \det
         \begin{pmatrix}
            \ \lambda_0-\alpha & \beta\\
            \beta^* & \lambda_1-\gamma
        \end{pmatrix} 
        = & (\lambda_0-\alpha)(\lambda_1-\gamma)-\abs{\beta}^2 \nonumber\\
        = & 0, 
    \end{align}
    where the following constraints on the parameters ensure that $M$ and $N$ are valid POVMs:
    \begin{align}
        \lambda_0, \lambda_1 \in [0,1]\\
        \alpha, \gamma \in [0,1]\\
        \min\{\alpha\gamma,(1-\alpha)(1-\gamma)\} \geq \abs{\beta}^2. 
    \end{align}

    Let $\lambda_0 = 1$ and $\lambda_1 = 0$. Then, Equation~\ref{eq:qubit_2_out_eq} reduces to 
    \begin{align}
        (1-\alpha)(-\gamma) = \abs{\beta}^2,
    \end{align}
    but since both $(1-\alpha), \gamma$ and $\abs{\beta}^2$ are positive, this equation is only solved when $\gamma = \beta = 0$. This gives
    \begin{align}
        M_0-N_0 =  &
        \begin{pmatrix}
            1 &0\\
            0 & 0
        \end{pmatrix}
        -
        \begin{pmatrix}
            \alpha &0\\
            0 & 0
        \end{pmatrix} \\
        = &
        \begin{pmatrix}
            1-\alpha &0\\
            0 & 0
        \end{pmatrix}
    \end{align}
    for any $\alpha \in [0,1]$. Conversely, when $\lambda_0=0$ and $\lambda_1 =1$, then $\alpha = \beta = 0$ and 
    \begin{align}
        M_0-N_0 = &
        \begin{pmatrix}
            0 &0\\
            0 & 1 
        \end{pmatrix}
        -
        \begin{pmatrix}
            0 &0\\
            0 & \gamma 
        \end{pmatrix} \\
        = &
        \begin{pmatrix}
            0 &0\\
            0 & 1-\gamma. 
        \end{pmatrix}
    \end{align}
    This shows that if one measurement is projective, the only possibility for forming a MEWC pair is if the other measurement is a biased version of the same measurement.
\end{proof}
Note that these two measurements only differ in the classical post-processing of a projective measurement. The measurement $M$ is a projective measurement, guaranteed to return $0$ if it measures state $0$, and $1$ if it measures state $1$. $N$, on the other hand, performs the same measurement, but with some probability of returning the wrong value.   
    
So far, every example studied has been commuting. Generally, however, the measurements do not have to commute in order for the upper bound to be saturated, as demonstrated in the following result: .

\begin{corollary}[Non-commuting measurement channel pairs that are MEWC]
    \label{cor:non-commuting}
     Let $\Phi_M,\Phi_N: \mathcal{L}(\hilbert{A}) \to \mathcal{L}(\hilbert{B})$ be a pair of measurement channels corresponding to the two 2-outcome measurements $M =\{M_i\}_{i=0}^{1}$ and $N = \{N_i\}_{i=0}^{1}$. There exists such pairs of channels where $M$ and $N$ do not commute that form MEWC pairs. For instance, let
     \begin{align}
           M_0 = & 
         \begin{pmatrix}
            \frac{3}{10} & \sqrt{\frac{1}{10}}\\
            \sqrt{\frac{1}{10}} &\frac{7}{10}
        \end{pmatrix}, \qquad  M_1 = I-M_0
\end{align}

\begin{align}
           N_0 = & 
           \begin{pmatrix}
            \frac{3}{10} & \sqrt{\frac{1}{10}}\\
            \sqrt{\frac{1}{10}} &\frac{7}{10}
        \end{pmatrix}, \qquad  N_1 = I-N_0. 
\end{align}
This is a MEWC pair. 
\end{corollary}
\begin{proof}
We have
    \begin{align}
           M_0-N_0 = & 
        \begin{pmatrix}
            \frac{1}{10} &0\\
            0&\frac{1}{5}
        \end{pmatrix}
        -
        \begin{pmatrix}
            \frac{3}{10} & \sqrt{\frac{1}{10}}\\
            \sqrt{\frac{1}{10}} &\frac{7}{10}
        \end{pmatrix}\\ 
        = &
        \begin{pmatrix}
            -\frac{1}{5}&-\sqrt{\frac{1}{10}} \\
            -\sqrt{\frac{1}{10}} & -\frac{1}{2}
        \end{pmatrix}, 
    \end{align}
    which can be diagonalised to 
    \begin{align}
    \begin{pmatrix}
        -\frac{7}{10} & 0 \\
         0 & 0
    \end{pmatrix}, 
     \end{align}
so the M-operator is rank $1$. Following Corollary \ref{corr:2_outcome_meas}, this means that we have a MEWC pair, showing that the MEWC condition does not entail that the measurements have to commute.
\end{proof} 
Note that the example in Corollary \ref{cor:non-commuting} above is the same discrimination task as in Figure~\ref{fig:m_plot}.
\subsection{Unitary channels that are MEWC pairs}
\label{ch:unitary_ch}
We will now present the details of the unitary channel example in Section~\ref{ch:motivation_ex_ME_bad}. For two unitaries $U$ and $V$ with corresponding quantum channels $\Phi_U$ and $\Phi_V$,~\cite{manna_maximally_2025} shows that the maximum probability of successfully discriminating  between them when using an input state with no entanglement is given by
\begin{align}
    P_{NE} = \frac{1}{2}\left[1 + \sqrt{1- \min{\abs{\operatorname{con}{\{e ^{i\theta_j} \}}}}^2}\right] \\ = \frac{1}{2}\left[1 + \sqrt{1- \abs{r_{\text{NE}}}^2}\right],
\end{align}
where $\{e ^{i\theta_j}\}_j$ are the eigenvalues of $U^{\dagger}V$, $\operatorname{con}{\{e ^{i\theta_j} \}}$ denotes all convex combinations of $\{e ^{i\theta_j}\}_j$ and we define $r_{\text{NE}} = \min{\abs{\operatorname{con}{\{e ^{i\theta_j} \}}}}$. The maximum success probability when using a maximally entangled input state is 
\begin{align}
    P_{ME} = \frac{1}{2}\left[1 + \sqrt{1- \frac{1}{d^2}\abs{\Tr(U^{\dagger}V)}^2}\right] \\ = \frac{1}{2}\left[1 + \sqrt{1- \frac{1}{d^2}\abs{\sum_{j=0}^{d-1} e^{i\theta_j}}^2}\right] \\
    = \frac{1}{2}\left[1 + \sqrt{1- \abs{r_{\text{ME}}} ^2}\right], 
\end{align}
where we have defined $ r_{\text{ME}} = \frac{1}{d}\sum_{j=0}^{d-1} e^{i\theta_j}$. 

The success probability in both cases is given by the eigenvalues of $U^{\dagger}V$, which are points on the complex unit circle, since this is a unitary. Therefore, the success probabilities can be understood by studying the convex hull of these points on the complex unit circle (see Figure~\ref{fig:unit_circle}.) 
\begin{figure}[H]
    \centering
        \includegraphics[width=\columnwidth]{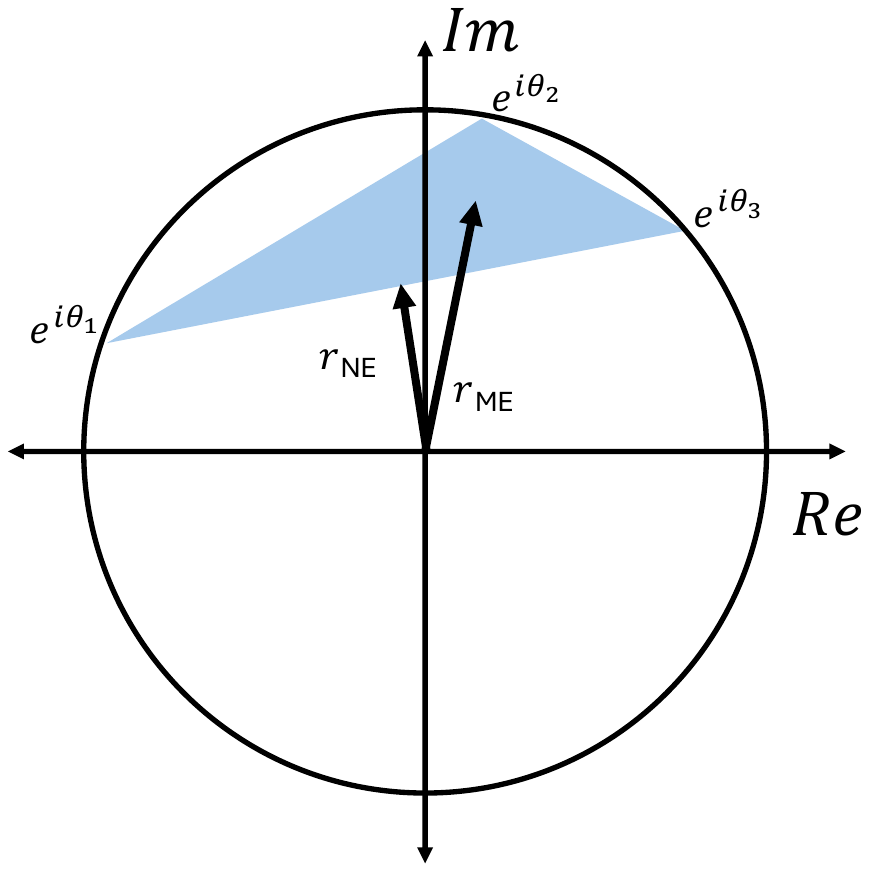}
        \caption{The convex hull of the eigenvalues of a unitary in dimension $d=3$. The discriminability with a maximally entangled input state depends on the centroid of the convex hull, described by $r_{ME}$, whereas with no entanglement the discriminability depends on $r_{NE}$ which corresponds to the minimum within the convex hull. This Figure is inspired by~\cite{manna_maximally_2025}.}
        \label{fig:unit_circle}
\end{figure}
If the centroid coincides with the origin, $P_{\text{ME}} = 1$, and the unitaries are perfectly distinguishable with a maximally entangled input state. A separable input state, on the other hand, yields perfect discrimination as long as the origin lies within the convex hull of the eigenvalues, as $r_{\text{NE}}$ represents the minimal distance from the origin to the polygon, which is zero whenever the polygon contains the origin. This highlights the results of~\cite{dariano_improved_2002}, which show that a separable state is always as good as a maximally entangled state for discrimination of unitaries. We also note that for qubit unitaries, which only have two eigenvalues, the centroid always coincides with the minimum inside the convex hull, so a maximally entangled input state is always optimal for discrimination of two qubit unitaries. Hence, any pair of qubit unitaries is a MEBC pair. 

Now, consider a pair of qudit unitaries $U$ and $V$ for which $W = U^{\dagger}V$ has $d-1$ equal eigenvalues and the final eigenvalue is at the complete opposite side of the complex unit circle. For example, if $W = U^{\dagger}V = \operatorname{diag}(-1,1,1,..,1)$, then $W$ has one eigenvalue at $e^{i\pi} = -1$ and $d-1$ eigenvalues at $e^{i0} = 1$, as demonstrated in Figure~\ref{fig:unit_circle_ex}.

\begin{figure}[H]
    \centering
        \includegraphics[width=\columnwidth]{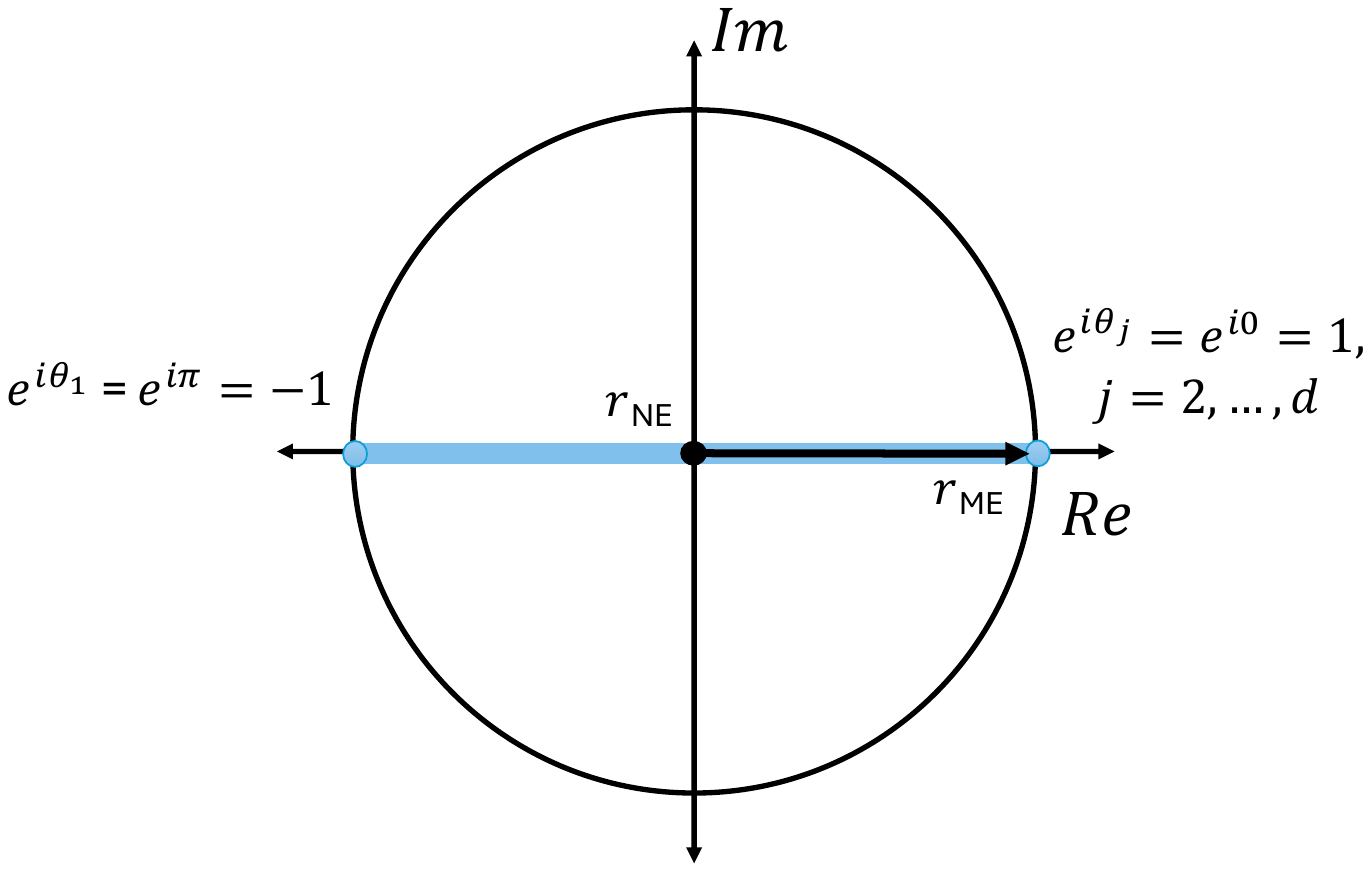}
        \caption{The eigenvalues of the unitary $W  = \operatorname{diag}(-1,1,1,..,1)$ on the complex unit circle. We see that the origin is within the convex hull of the eigenvalues, but the centroid approaches one as the dimension $d$ increases.} 
        \label{fig:unit_circle_ex}
\end{figure}

In the geometrical picture described above, we see that $U$ and $V$ are perfectly distinguishable with a separable input state since the convex hull of the eigenvalues includes the origin. However, the centroid approaches the edge of the complex unit circle in the large-$d$ limit, so a maximally entangled input state is maximally suboptimal in the limit $d\to \infty$:
    \begin{align}
      \lim_{d \rightarrow \infty}P_{ME} = & \lim_{d \rightarrow \infty}\frac{1}{2}\left[1 + \sqrt{1- \frac{1}{d^2}\abs{\sum_{j=0}^{d-1}e^{i\theta_j}}^2}\right]  
    \\ = & \lim_{d \rightarrow \infty}\frac{1}{2}\left[1 + \sqrt{1- \frac{(d-2)^2}{d^2}}\right] \\
    = & \lim_{d \rightarrow \infty}\frac{1}{2}\left[1 + \sqrt{1- \left(1-\frac{2}{d}\right)^2}\right] = \frac{1}{2}.
\end{align}

With a separable input state, on the other hand, the probability of correctly discriminating the channels is given by 
    \begin{align}
        P_{prod} = \frac{1}{2}\big[1 + \sqrt{1- \min{\abs{\text{con}{\{e ^{i\theta_j} \}}}}^2}\big] \\ = \frac{1}{2}\big[1 + \sqrt{1- 0} \big] = 1 
    \end{align}
for all system sizes $d$. 

The strategies in the two cases are equivalent to the ones described in the introduction.
\section{Discussion and conclusion}
In conclusion, this work considers the question of when entanglement reduces instead of enhances discriminability between pairs of quantum channels, countering the intuition one might have that entanglement is always useful for discrimination task. We define the terms Maximal Entanglement Best Case and Maximal Entanglement Worst case pairs to mean pairs of quantum channels for which the inequality 
\begin{align}
    \norm{\Delta_{\Phi}}_{\text{ME}} \leq \norm{\Delta_{\Phi}}_{\diamond} \leq d\,\norm{\Delta_{\Phi}}_{\text{ME}} 
\end{align} 
is saturated in the lower and upper bound, respectively, identifying channel discrimination tasks where entanglement acts as a resource and when it does not. 

We present conditions for and examples of these two cases, in particular focusing on the cases when entanglement is not useful for discrimination. Through plotting the maximal success probability of discrimination as a function of the entanglement entropy of the input state, we present a finer quantitative analysis of the effect of entanglement on different channel discrimination tasks. These plots demonstrate our analytical results showing that there exists pairs of channels for which a separable input state is optimal and adding entanglement strictly reduces discriminability. 

Through this analysis, we identify the M-operator as an important quantity yielding information about the optimal input state without performing the full discrimination task which demands optimisation. We can, for instance, easily check whether entanglement is useful or not, and we can identify the subspace within the input space where any optimal input state lives. It is also easy to check if a maximally entangled input is optimal.

A natural extension of this work would be to consider the discrimination of more then two quantum channels. Here, we have only considered discrimination of pairs of quantum channels, in which case we observe that the maximum advantage that could be achieved using an optimal input state (as opposed to a maximally entangled input state) is bounded by the system size $d$. This arises from the fact that the maximum success probability of discrimination for pairs of channels (or states) depends on the diamond norm and the trace norm. For discrimination of $N$ quantum channels, the success probability is no longer expressed via these norms, and there is a chance that the maximally entangled input state might perform significantly worse in some cases. Hence, a natural next step from this work would be to characterise the cases when the maximally entangled input state is a good or a bad input state for discrimination of $N$ quantum channels, and to investigate whether the possible disadvantage in this case is unbounded.

Furthermore, it would also be interesting to investigate the MEWC conditions for other classes of channels, allowing us to identify more examples of MEWC pairs.

\section*{Acknowledgments}
We thank Seiseki Akibue and Jonas Gran Melandsør for fruitful discussions.
MTQ is supported by the Agence Nationale de la Recherche (ANR) through the JCJC programme under grant number ANR-25-CE47-6396-01-HOQO-KS. 




%

\appendix
\label{ch:appendix}
\section{Details about the plots of success probability as a function of entanglement entropy}
Since the main focus of this work is to study the way entanglement affects discrimination, we are interested in quantifying how the discriminability varies when the amount of entanglement changes. For qubits, the amount of entanglement is fully determined by one parameter, something which allows for studying this relation through plots like the ones shown in Section \ref{ch:motivation_ex}. In this section, we will discuss how these plots are made. 

A two-qubit input state can be parametrised as
\begin{align}
    \ket{\Psi_{\theta}} =\cos{\theta}\ket{00} + \sin{\theta}\ket{11} \in \hilbert{A}\otimes \hilbert{A'}, 
\end{align}
where $\theta = 0$ gives the pure state $\ket{\Psi_{0}}=\ket{00}$ and $\theta=\frac{\pi}{4}$ gives a maximally entangled state $\ket{\Psi_{\frac{\pi}{4}}}=\frac{1}{\sqrt{2}}\left(\ket{00}+ \ket{11}\right)$. Notice that the entanglement entropy of the state, given by 
\begin{align}
    S(\ketbra{\Psi_\theta})= &-\operatorname{Tr}[\rho_A\log(\rho_A)] \\
    =& -\cos^2{\theta}\log\cos^2{\theta}-\sin^2{\theta}\log\sin^2{\theta}, \nonumber
\end{align}
where $\rho_A = \operatorname{Tr}_{A'}\ketbra{\Psi_{\theta}}$, and therefore also the entanglement entropy, only depends on the parameter $\theta$.   

However, the classification of a state by its Schmidt coefficients is only unique up to local unitaries, which means that any state $U_A \otimes U_{A'} \ket{\Psi_{\theta}}$ for arbitrary unitaries $U_A, U_{A'}$ has the same entanglement entropy as $\ket{\Psi_{\theta}}$. Therefore, in order to find the maximum success probabiloity for a given, fixed amount of entanglement, we need to find the unitaries $U_A, U_{A'}$ that give the highest success probability for a given $\theta$. This discrimination strategy is shown in Figure \ref{fig:med_unitaries}. 
\begin{figure}[H]
        \centering
        \includegraphics[width=\linewidth]{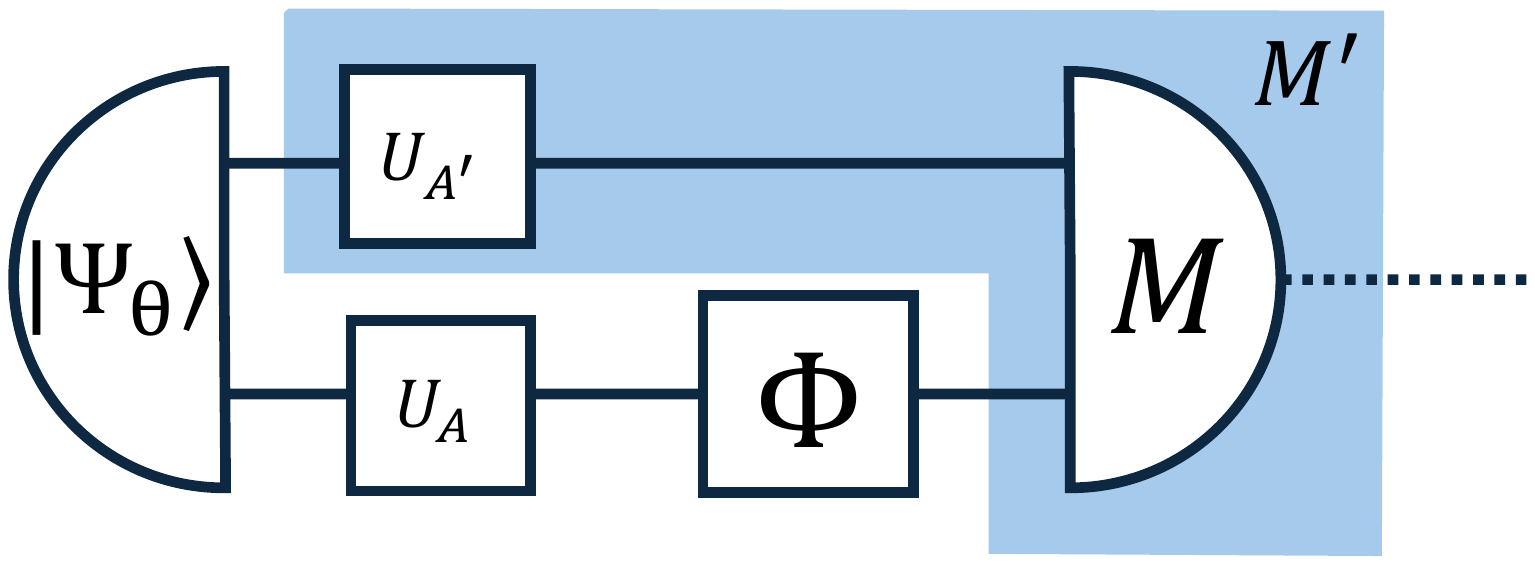}
        \caption{Quantum channel discrimination with parametrised two-qubit input states $\ket{\Psi_{\theta}}$. To find the maximum success probability for a given amount of entanglement as fixed by $\theta$, we need to find the best unitaries $U_A, U_{A'}$. However, we see that $U_{A'}$ can be absorbed in the measurement $M$, and since we optimise over all measurements, we do not have to consider the unitary $U_{A'}$.}
        \label{fig:med_unitaries}
\end{figure}
We notice, however, that the unitary $U_{A'}$ can be absorbed into the measurement $M$ to give a different measurement $M'$, and since the measurement is optimised over in any case, the unitary $U_{A'}$ does not make any difference. The unitary $U_A$, on the other hand, does in general change the discrimination success probability. This is an issue that has been treated differently in the different individual cases.

Firstly, when the input state is maximally entangled, the identity $(U \otimes I)\ket{\phi^+} = (I \otimes U^T)\ket{\phi^+}$ shows that applying a unitary to one part of a maximally entangled state is equivalent to applying its transpose to the opposite part. In this way, $U_A$ too, just as $U_{A'}$ is absorbed in the measurement and automatically optimised. This is the reason why it suffices to show that one maximally entangled input in order to know that any maximally entangled input state is optimal. 

Furthermore, in some cases we can argue that the unitary $U_A$ is of no importance by showing that the set of channels we aim to discriminate are jointly unitary covariant, a property which is defined as follows: 
\begin{definition}[Joint unitary covariance]
    A set of channels $\{\Phi_i\}_i$ is said to by jointly unitary covariant if for all $i$, 
    \begin{align}
        \Phi_i(U\rho U^{\dagger}) = V_U \Phi_i V_U^{\dagger}, 
    \end{align}
    where $U$ and $V_U$ are unitaries. 
\end{definition} 

If this is the case, applying a unitary $U_A$ before the channel is equivalent to applying a different unitary $V_{U_A}$ after the channel, but in this case it can be absorbed in the measurement as previously argued, and we can therefore disregard it altogether.

For the case of the Werner-Holevo channels, it can be shown that applying a unitary $U$ before the channel is always equivalent to applying a unitary $V_U = \Bar{U}$ after the channel, where $\Bar{U}$ is the complex conjugate of $U$\footnote{Note that this is also the case for the broader class of Werner-Holevo channels of the form $\Phi_{WH} = p \Phi_{+} + (1-p) \Phi_- $ with $p \in (0,1)$, and not just the extremal ones studied here.}: 
\begin{align}
    \Phi_{\pm}(U \rho U^{\dagger}) = &\operatorname{Tr(U \rho U^{\dagger})I \pm (U \rho U^{\dagger})^T} \\
    = & \operatorname{Tr}(\rho)I \pm U^{\dagger T} \rho^T U^T \\
    = & \operatorname{Tr}(\rho)\Bar{U} \Bar{U^{\dagger}} \pm \Bar{U} \rho^T \Bar{U^{\dagger}} \\
     = & \Bar{U} (\operatorname{Tr}(\rho)I \pm  \rho^T ) \Bar{U^{\dagger}} \\
     = & \Bar{U} \Phi_{\pm}(\rho) \Bar{U^{\dagger}}. 
\end{align}
Therefore, we can safely use the input state $\ket{\Psi_{\theta}}$ directly without being concerned with the unitary $U_A$ in this case. 

For the four Pauli channels presented in Section \ref{ch:motivation_ex}, this is not the case. However, in this case, we have the following analytical expression for the maximum success probability of discrimination as a function of $\theta$:
\begin{align}
    p^* = \frac{1}{2}(1+ \sin(2 \pi)),
\end{align} 
This can be understood as follows. The four bipartite output states resulting from applying the four Pauli channels are given by
\begin{align}
     \ket{\psi_X} &= (I\otimes X) \ket{\Psi_{\theta}} = \cos(\theta) \ket{01} + \sin{\theta}\ket{10}  \\
    \ket{\psi_Y} &= (I\otimes Y) \ket{\Psi_{\theta}} = \cos(\theta) \ket{01} - \sin{\theta}\ket{10}  \\
    \ket{\psi_Z} &= (I\otimes Z) \ket{\Psi_{\theta}} = \cos(\theta) \ket{00} - \sin{\theta}\ket{11}  \\
    \ket{\psi_I} &= (I\otimes I) \ket{\Psi_{\theta}}=  \cos(\theta) \ket{00} + \sin{\theta}\ket{11} . 
\end{align}
These can be split into two subsets of output states $\{\ket{\psi_Z}, \ket{\psi_I}\}$ and $\{\ket{\psi_X}, \ket{\psi_Y}\}$ that can be perfectly distinguished between since they live in orthogonal subspaces $A = \operatorname{Span}\{\ket{00}, \ket{11}\}$ and  $B = \operatorname{Span}\{\ket{01}, \ket{10}\}$, respectively. Within each subspace, we have a dicrimination task between a pair of quantum channels, and the Holevo-Helstrom theorem gives: 
\small
\begin{align}
     p^*_A = \frac{1}{2}(1+ \sqrt{1- \braket{\psi_I}{\psi_Z}}) = \frac{1}{2}(1+ \sin(2 \pi)), \\
   \small p^*_B = \frac{1}{2}(1+ \sqrt{1- \braket{\psi_X}{\psi_Y}}) = \frac{1}{2}(1+ \sin(2 \pi)), 
\end{align}
\normalsize
so an upper bound on the success probability for discriminating the four channels can be found by assuming perfect discrimination between the two orthogonal subspaces and optimal discrimination within the subspaces: 
\begin{align}
    p^* =&  p(\text{subspace A})p(\text{guess correct in subspace A}) \\ & +  p(\text{subspace B})p(\text{guess correct in subspace B})\\
     =& \frac{1}{2}\frac{1}{2}(1+ \sin(2 \pi)) + \frac{1}{2}\frac{1}{2}(1+\sin(2 \pi))\\
     = & \frac{1}{2}(1+ \sin(2 \pi)).  
\end{align}
This function perfectly matches the optimised maximal success probability given in Figure \ref{fig:pauli_plot}, so we see that the bound is attained. Furthermore, we also get the same curve when heuristically optimising over a parametrised unitary as explained below, so it seems like the unitary $U_A$ does not affect discrimination in this case either. 

For the types of channels considered so far, it has turned out that the unitary $U_A$ does not make any difference. This is, however, not the case for the class of measurement channels in general. This was solved by heuristically optimising over a parametrisation of the unitary given by
\begin{align}
    U = 
    \begin{pmatrix}
        e^{(i \alpha)\cos(\gamma)} & e^{(i \beta)\sin(\gamma)} \\
        -e^{(-i \beta)\sin(\gamma)} & e^{(-i \alpha)\cos(\gamma)}
    \end{pmatrix}. 
\end{align} 
Although we have no guarantee that this optimisation finds the globally optimal unitary, the results seem to be consistent with what we expect in the cases studied. It also seems to be stable across different initialisations. 

For the particular example considered in \ref{fig:m_plot}, the optimised unitary $U_A$ is given by 
\begin{align}
    \begin{pmatrix}
        0.85 & 0.53 \\
        -0.53 & 0.85
    \end{pmatrix}. 
\end{align}

\end{document}

%% file: 0_MTQ_macros.tex
\newcommand{\1}{\mbox{1}\hspace{-0.25em}\mbox{l}}

\makeatletter 
\newcommand{\doublewidetilde}[1]{{%
  \mathpalette\double@widetilde{#1}%
}}
\newcommand{\double@widetilde}[2]{%
  \sbox\z@{$\m@th#1\widetilde{#2}$}%
  \ht\z@=.9\ht\z@
  \widetilde{\box\z@}%
}
  
\makeatother

\newtheorem{definition}{Definition}
\newtheorem{theorem}{Theorem}
\newtheorem{lemma}{Lemma}
\newtheorem{proposition}{Proposition}
